%% file: paper.tex
\pdfoutput=1

\newif\ifdraft\draftfalse 
\newif\ifanon\anonfalse   
\newif\iffull\fullfalse   
\newif\iflongrefs\longrefsfalse 
\newif\ifbackref\backreffalse 
\newif\ifsooner\soonerfalse
\newif\iflater\laterfalse
\newif\ifcamera\cameratrue 
\newif\ifcheckpagebudget\checkpagebudgetfalse
\newcommand{\xxx}{}
\makeatletter \@input{texdirectives.tex} \makeatother

\ifcamera
\documentclass[acmsmall]{acmart}
\else
\ifanon
\documentclass[acmsmall,review,anonymous]{acmart}
\else
\documentclass[acmsmall,review=false,screen=true]{acmart}
\fi
\fi

\acmJournal{PACMPL}
\acmYear{2018}
\acmMonth{1}
\acmVolume{2}
\acmNumber{POPL}
\acmPrice{}
\acmArticle{65}
\acmMonth{1}
\acmDOI{10.1145/3158153}
\ifcamera\else\copyrightyear{2017}\fi 

\ifanon
\setcopyright{none}             
\else
\setcopyright{rightsretained}   
\fi

\makeatletter
\def\@copyrightpermission{\ifcamera\\\\\\\fi This work is licensed under a \href{https://creativecommons.org/licenses/by/4.0/}{Creative Commons Attribution 4.0 International License}}
\makeatother

\ifcamera\else
\makeatletter
\def\@authorsaddresses{}
\fancypagestyle{firstpagestyle}{%
  \fancyhf{}%
  \renewcommand{\headrulewidth}{\z@}%
  \renewcommand{\footrulewidth}{\z@}%
    \fancyhead[L]{\ifanon\ACM@linecountL\fi}%
}
\fancypagestyle{standardpagestyle}{%
  \fancyhf{}%
  \renewcommand{\headrulewidth}{\z@}%
  \renewcommand{\footrulewidth}{\z@}%
    \fancyhead[LE]{\ifanon\ACM@linecountL\fi\@headfootfont\thepage}%
    \fancyhead[RO]{\@headfootfont\thepage}%
    \fancyhead[RE]{\@headfootfont\@shortauthors}%
    \fancyhead[LO]{\ifanon\ACM@linecountL\fi\@headfootfont\shorttitle}%
    \fancyfoot[RO,LE]{}%
}
\def\@mkbibcitation{}

\makeatother
\fi

\usepackage{afterpage}
\usepackage{color}
\usepackage{amssymb,amsmath,amsfonts}
\usepackage{extarrows}
\usepackage{version}
\usepackage{xspace}
\usepackage{graphicx}
\usepackage{listings}
\usepackage{wrapfig}
\usepackage{ifpdf}
\usepackage{semantic}
\usepackage{mathpartir} 
\newcommand{\tightlabel}[2]{\makebox[1pt][l]{\raisebox{+1em}{\hspace{0.5em}(#1)}}}
\newcommand{\tightlabelup}[2]{\makebox[1pt][l]{\raisebox{+2.4em}{\hspace{0.5em}(#1)}}}
\newcommand\citepos[1]{\citeauthor{#1}'s\ \citeyear{#1}}
\usepackage{amsthm}
\usepackage{stmaryrd}
\usepackage{subfigure}
\usepackage{multirow}
\usepackage{proof}
\usepackage{enumitem}
\ifcamera
\usepackage{flushend} 
\fi
\definecolor{darkblue}{rgb}{0.0,0.0,0.3}

\usepackage{hyperref}
\hypersetup{breaklinks=true}

\include{minimalpreamble}

\newcommand\fstar{F$^\star$\xspace}

\newcommand\MSTcalculus{$\lambda_{\text{MST}}$\xspace}
\newcommand\rMSTcalculus{$\lambda_{\text{MST}_{b}}$\xspace}

\definecolor{dkblue}{rgb}{0,0.1,0.5}
\definecolor{dkgreen}{rgb}{0,0.4,0}
\definecolor{dkred}{rgb}{0.6,0,0}
\definecolor{dkpurple}{rgb}{0.7,0,1.0}
\definecolor{purple}{rgb}{0.9,0,1.0}
\definecolor{olive}{rgb}{0.4, 0.4, 0.0}
\definecolor{teal}{rgb}{0.0,0.4,0.4}
\definecolor{azure}{rgb}{0.0, 0.5, 1.0}
\definecolor{gray}{rgb}{0.5, 0.5, 0.5}
\definecolor{dkgray}{rgb}{0.3, 0.3, 0.3}

\newcommand{\comm}[3]{\ifcheckpagebudget\else\ifdraft{{\color{#1}[#2: #3]}}\fi\fi}
\newcommand{\da}[1]{\comm{dkblue}{Danel}{#1}}
\newcommand{\nik}[1]{\comm{dkpurple}{Nik}{#1}}
\newcommand{\ch}[1]{\comm{teal}{CH}{#1}}
\newcommand{\cf}[1]{\comm{dkgreen}{Cedric}{#1}}

\newcommand\statet{\ensuremath{\mathsf{state}}}
\newcommand\unit{\ensuremath{\mathsf{unit}}}
\newcommand\boolt{\ensuremath{\mathsf{bool}}}

\newcommand\true{\ensuremath{\mathsf{true}}}
\newcommand\false{\ensuremath{\mathsf{false}}}

\newcommand\rel{\ensuremath{\mathsf{rel}}}

\newcommand\mst[3]{\ensuremath{{\color{dkblue}\mathsf{MST}}~#1~(#2)~(#3)}}
\newcommand\bmst[4]{\ensuremath{{\color{dkblue}\mathsf{MST}}_{#1}~#2~(#3)~(#4)}}
\newcommand\pure[3]{\ensuremath{{\color{dkblue}\mathsf{Pure}}~#1~(#2)~(#3)}}

\newcommand\inl[1]{\ensuremath{\mathsf{inl}~#1}}
\newcommand\inr[1]{\ensuremath{\mathsf{inr}~#1}}

\newcommand\return[1]{\ensuremath{\mathsf{return}~#1}}
\newcommand\breturn[2]{\ensuremath{\mathsf{return}_{#1}~#2}}
\newcommand\bind[3]{\ensuremath{\mathsf{bind} ~ #1 \leftarrow #2 ~ \mathsf{in} ~ #3}}

\newcommand\stput[1]{\ensuremath{\mathsf{put}~#1}}
\newcommand\stget{\ensuremath{\mathsf{get}}}
\newcommand\witness[1]{\ensuremath{\mathsf{witness}~#1}}
\newcommand\recall[1]{\ensuremath{\mathsf{recall}~#1}}
\newcommand\witnessed[1]{\ensuremath{\mathsf{witnessed}~#1}}
\newcommand\stable[1]{\ensuremath{\mathsf{stable}~#1}}

\newcommand\bstput[2]{\ensuremath{\mathsf{put}_{#1}~#2}}
\newcommand\bstget[1]{\ensuremath{\mathsf{get}_{#1}}}
\newcommand\bwitness[2]{\ensuremath{\mathsf{witness}_{#1}~#2}}
\newcommand\brecall[2]{\ensuremath{\mathsf{recall}_{#1}~#2}}

\newcommand\reify[1]{\ensuremath{\mathsf{reify}~#1}}
\newcommand\reflect[1]{\ensuremath{\mathsf{reflect}~#1}}
\newcommand\coerce[1]{\ensuremath{\mathsf{coerce}~#1}}

\newcommand\pre[0]{\ensuremath{\varphi_{\text{pre}}}}
\newcommand\post[0]{\ensuremath{\varphi_{\text{post}}}}

\newcommand\pmatch[5]{\ensuremath{\mathsf{pmatch}~#1~\mathsf{as}_{#2}~(#3,#4)~\mathsf{in}~e}}

\newcommand\caseof[6]{\ensuremath{\mathsf{case}~#1~\mathsf{of}_{#2}~(\mathsf{inl}(#3)~\mathsf{in}~#4 ; \mathsf{inl}(#5)~\mathsf{in}~#6)}}

\newcommand\fun[3]{\ensuremath{\lambda #1{:}#2.~#3}}

\newcommand\wf{\ensuremath{\mathsf{wf}}}

\newcommand*{\EG}{e.g.,\xspace}
\newcommand*{\IE}{i.e.,\xspace}
\newcommand*{\ETAL}{et al.\xspace}

\newcommand\surl[1]{{\small\url{#1}}}
\newcommand\ssurl[1]{{\scriptsize\url{#1}}}

\begin{document}

\title{Recalling a Witness}
\subtitle{Foundations and Applications of Monotonic State}

\xxx{}

\ifanon
\author{}
\else
\author{Danel Ahman}
\affiliation{
  \ifcamera\institution{Inria}\city{Paris}\country{France}
  \else\institution{Inria Paris}\fi}
\author{C\'{e}dric Fournet}
\affiliation{\institution{Microsoft Research}
  \ifcamera\country{UK}\fi}
\author{C\u{a}t\u{a}lin Hri\c{t}cu}
\affiliation{
  \ifcamera\institution{Inria}\city{Paris}\country{France}
  \else\institution{Inria Paris}\fi}
\author{Kenji Maillard}
\affiliation{
  \ifcamera\institution{Inria and ENS}\city{Paris}\country{France}
  \else\institution{Inria Paris and ENS Paris}\fi}
\author{Aseem Rastogi}
\affiliation{\institution{Microsoft Research}
  \ifcamera\country{India}\fi}
\author{Nikhil Swamy}
\affiliation{\institution{Microsoft Research}
  \ifcamera\country{USA}\fi}
\makeatletter
\renewcommand{\@shortauthors}{Ahman~\ETAL}
\makeatother
\fi

\begin{abstract}
  We provide a way to ease the verification of programs
  whose state evolves monotonically.
  The main idea is that
  a property \emph{witnessed} in a prior state can be soundly \emph{recalled}
  in the current state, provided (1)~state evolves according to a
  given preorder, and (2)~the property is preserved by this preorder.
  In many scenarios, such monotonic reasoning yields concise
  modular proofs, saving the need for explicit program invariants.
  We distill our approach into the \emph{monotonic-state monad}, a
  general yet compact interface for Hoare-style reasoning
  about monotonic state in a dependently typed language.
  We prove the soundness of the monotonic-state monad and use it as a
  unified foundation for reasoning about monotonic state in the
  \fstar{} verification system.
  Based on this foundation, we build libraries for various mutable data
  structures like monotonic references and apply these libraries at
  scale to the verification of several distributed applications.
\end{abstract}

\ifcamera

\fi


\maketitle

\section{Introduction}
\label{sect:introduction}

Functional programs are easier to reason about than stateful programs,
inasmuch as properties proven on pure terms are preserved by
evaluation.
In contrast, properties of imperative programs generally depend on
their evolving state, 
e.g., a counter initialized to zero may later contain a strictly
positive number, requiring properties that depend on this counter to
be revised.

To reign in the complexity of reasoning about ever-changing state,
verification can be structured using \emph{stateful invariants},
i.e., predicates capturing properties that hold in every program
state. Defining invariants and proving their preservation are the
bread and butter of verification; and many techniques and tools have
been devised to aid these tasks.
A prominent such example is separation
logic~\citep{Reynolds2002SL,ishtiaq01sep}, which offers a way to
compose invariants according to the shape of mutable data.

\da{The above is only true of Reynolds style vanilla separation logic.}
\nik{added footnote; ok?}

Whereas separation logic is concerned primarily with \emph{spatial}
properties,\footnote{Recent variants of separation logic consider
resources more abstract than just heap locations. In such settings, in
addition to describing spatial properties of resources, one can encode
certain kinds of temporal properties. We compare our work to these
modern variants of separation logic in \S\ref{sec:relatedwork},
focusing here on more familiar program logics for reasoning about
heaps.}  program verification also makes use of \emph{temporal}
properties that control how the state may evolve.
For instance, consider a program with a counter that can only be increased.
Reading $n$ as its current value should allow one to conclude that its
value will remain at least $n$, irrespective of state updates. In
turn, this may be used to reason about the generation of unique
identifiers from fresh counter values.
One may hope to recover some of the ease of reasoning about pure
programs in this setting, at least for properties that are preserved
by counter increments.
Capturing this intuition formally and using it to simplify the
verification of stateful programs are the main goals of this paper.

\subsection{Stateful Invariants vs. Monotonic State and Stable Predicates}

Consider a program written against a library that operates
on a mutable set. The library provides abstract operations
to read the current value of the set and test whether an element is in
the set ($\in$).
Importantly, its only operation to mutate the state is `\ls$insert$',
which adds an element into the set.
From this signature alone, one should be able to conclude
that an element observed in the set will remain in it. However,
proving this fact is not always easy in existing program logics.
To illustrate our point, consider trying to prove the assertion
at the end of the following program:

\begin{center}
\begin{tabular}{c}
\begin{lstlisting}
insert v; complex_procedure(); assert (v $\in$ get())
\end{lstlisting}
\end{tabular}
\end{center}

\vspace{0.05cm}

\da{We need to revise this example, or at least its explanation, because 
      as pointed out by our reviewers, \ls$inv$ can be framed across 
      complex procedure even if it does inserts.}
      
\nik{Really? How so? Do you mean in the PCM-based setting?
If so, then I think the footnote I added should take care of it, I
hope.}

In a Floyd-Hoare logic, one may prove
that our code maintains a stateful invariant
on the set,
i.e., prove the Hoare triple \ls@{inv} complex_procedure() {inv}@,
where \ls@inv s@ is defined as \ls@ v $\in$ s@.
%
One may rely on separation logic to conduct this proof, for instance
by framing \ls$inv$ across \ls$complex_procedure$ provided it does not
operate on the set.
However, when \ls$complex_procedure$ also inserts elements, one must
carry \ls$inv$ through its proof and reason about the effect of these
insertions to confirm that they preserve \ls$inv$.
This quickly becomes tedious, e.g., showing that the set also
retains some other element \ls$w$ requires another adjustment to the
proof of \ls$complex_procedure$.
Besides, although this Hoare triple suffices to prove our assertion, it
does not by itself ensure that \ls$v$ remains in the mutable set
throughout execution. For example, \ls$v$ could have been temporarily
removed and then reinserted.
While detailed stateful invariants are often unavoidable, they are
needlessly expensive here: knowing that elements can only be inserted,
we would like to conclude that \ls@v $\in$ s@ is stable.

\paragraph{Our Solution: Monotonic State and Stable Predicates.}
Let $\leq$ be a preorder (that is, a reflexive-transitive relation)
on program states, or a fragment thereof.
We say that a program is \emph{monotonic} when its state evolves according
to the preorder (that is, we have $s_0 \leq s_1$ for any two successive states),
and that a predicate on its state is \emph{stable} when it is preserved by the preorder
(that is, when for all states $s_0$ and $s_1$, we have \ls@p $s_0$ /\ $s_0 \leq s_1$ ==>p $s_1$@).
We outline an extension of Hoare-style program logics in this setting:
we restrict any state-changing operations so that they
conform with~$\leq$, and we extend the logic and the underlying language with the following
new constructs:
\begin{enumerate}[leftmargin=0.7cm]
\item A logical capability, \ls$witnessed$, that turns a
  predicate \ls$p$ on states into a \emph{state-independent} proposition \ls$witnessed p$, expressing both that i) 
  \ls$p$ was observed in some past state and ii) it will hold in every future state, including the program's current state.

\item A weakening principle, \ls$(forall s. p s ==> q s) ==> (witnessed p ==> witnessed q)$.

\item Two actions: \ls$witness p$, to establish \ls$witnessed p$,
  given that \ls$p$ is stable and holds in the current state; and
  \ls$recall p$, to (re)establish \ls$p$ in the current state, given
  that \ls$witnessed p$ holds.
\end{enumerate}

Continuing with our example, one may pick set inclusion as the
preorder and check that the only mutating operation, \ls$insert$,
respects it.
Then, preserving (i.e., framing) stable properties across state
updates is provided by the logic whenever explicitly
requested using \ls$witness$ and \ls$recall$.
For instance, we can revise our example program as shown below
to prove the assertion:
\da{v remaining in the set (commented out) is independent of adding witness and recall}

\begin{center}
\begin{tabular}{c}
\begin{lstlisting}
insert v; $\color{dkred}\mathsf{witness~inv};$ complex_procedure(); $\color{dkred}\mathsf{recall~inv};$ assert (v $\in$ get())
\end{lstlisting}
\end{tabular}
\end{center}

\vspace{0.05cm}

\noindent Crucially, \ls$witness inv$ yields the postcondition
\ls$witnessed inv$, which, being a pure, \emph{state-independent} proposition,
is trivially maintained across \ls$complex_procedure$ without the need to
analyze its definition.
With \ls$recall inv$, we recover \ls$inv$ of the current state without
having to prove that it is related to the state in which \ls$inv$ was
witnessed, since this follows from the stability
of \ls$inv$ with respect to~$\subseteq$.
We could also prepend \ls$insert w$ to this code, and still would not
need to revisit the proof of \ls$complex_procedure$ to establish that
the final state contains \ls$w$---we only need to insert the
corresponding \ls@witness (fun s -> w $\in$ s)@ and
\ls@recall (fun s -> w $\in$ s)@ operations into our program.

\subsection{Technical Overview and Contributions}

A point of departure for our work is \citepos{fstar-pldi13} proposal
to reason about monotonic state using
a variant of the \ls$witness$ and \ls$recall$ primitives.
Following their work, the \fstar programming language~\citep{mumon}
has embraced monotonicity in the design of many of its verification
libraries. Although these libraries have been
founded on ad hoc axioms and informal meta-arguments,
%
they are used extensively in several large-scale projects, including
verified efficient cryptographic libraries \cite{lowstar, haclstarccs}
%
and a partially-verified implementation of TLS \cite{everest, record}.
The theory of monotonic state developed here provides a new, unifying
foundation accounting for all prior, ad hoc uses of monotonicity
in \fstar, while also serving as a general basis to the
verification of monotonic properties of imperative programs
in other systems.
Our contributions include the following main points.

\vspace{-0.025cm}

\paragraph{The Monotonic-State Monad (\autoref{sec:mst}).}
We propose \ls$MST$, the {\em monotonic-state monad}, and its encoding
within a sequential, dependently typed language like \fstar.
\ch{also ``encoding'' seems off here,
  all the interesting parts of MST are axioms, not encoded}
Its interface is simpler and more general
than~\citepos{fstar-pldi13} with just four actions for
programming and reasoning about global, monotonic state
(\ls$get$, \ls$put$, \ls$witness$, and \ls$recall$), together with a
new logical connective (\ls$witnessed$) for turning predicates on states
into state-independent propositions.

\vspace{-0.025cm}

\paragraph{Formal Foundations of the Monotonic-State Monad (\autoref{sect:metatheory}, \autoref{sect:reify}).}
We investigate the foundations of the monotonic-state monad
by designing a sequent calculus for a first-order logic with the \ls$witnessed$
connective, proving it consistent via cut admissibility. We use this to
prove the soundness of a Hoare-style logic for a core dependently typed lambda
calculus augmented with the \ls$MST$
interface. We present this in two stages. First, we focus on an abstract
variant of \ls$MST$ and prove our Hoare logic sound in the sense of
total correctness (\autoref{sect:metatheory}).
Next, we show how to soundly reveal the representation
of \ls$MST$ computations as pure state-passing functions, while carefully ensuring
that the preorder enforced by \ls$MST$ is not violated.
This pure representation can be used to conduct relational proofs
(e.g., showing noninterference) for programs with monotonic state, and
to enable clients to safely extend interfaces based on \ls$MST$
with new, preorder-respecting actions (\autoref{sect:reify}).

\ch{Should better explain the limitations of \autoref{sect:reify},
  in preparation for a proper solution to the problem.}

\vspace{-0.025cm}

\paragraph{Typed Heaps and Three Flavors of References  (\autoref{sec:references} and \autoref{sec:mref}).}
Using \ls$MST$, we encode more convenient forms of mutable state,
including 
heaps with dynamic allocation and
deallocation. Our model supports both untyped references \ls$uref$
with strong (non-type-preserving) updates, as well as typed references
(\ls$ref t$) with only weak (type-preserving) updates.
Going further, we use \ls$MST$ to program monotonic typed references
(\ls$mref t rel$) that allow us to select a preorder for each reference separately,
%
with \ls$ref t$ just a special case of monotonic references with a trivial
preorder.
As such, programmers can opt-in to monotonicity whenever they allocate
a reference, and retain the generality of non-monotonic state and
stateful invariants wherever needed.

\ifcamera
\vspace{-0.1cm}
\else
\vspace{-0.025cm}
\fi

\paragraph{Secure File-Transfer, a First Complete Example Application (\autoref{sec:action}).}
Based on monotonic references, we verify a secure
file-transfer application, illustrating several practical uses
of monotonicity: ensuring safe memory initialization; modeling
ghost distributed state as an append-only log of messages;
and the interplay between stateful invariants,
refinement types, and stable predicates. This 
application provides a standalone illustration of the essential use of
monotonicity in a larger scale \fstar verification project targetting
the main streaming authenticated encryption construction used in TLS 1.3~\cite{record}.

\ifcamera
\vspace{-0.1cm}
\else
\vspace{-0.025cm}
\fi

\paragraph{Ariadne, Another Application to State Continuity (\autoref{sec:case-studies.ariadne}).}
We also develop a new case study based on a recent protocol
\cite{StrackxP16} that ensures the state continuity of
hardware-protected sub-systems (such as TPMs and SGX
enclaves) in the presence of multiple crashes and restarts.
Our proof consists of programming a ``ghost state machine'' to keep
track of the progress of the protocol. It illustrates the combination
of ghost state and monotonicity as an effective style for verifying
distributed algorithms, following the intuition that in a
well-designed protocol, the logical properties carried by every
message must be stable.  While we do not address concurrency in
this paper, this case study also illustrates how monotonic state may be
usefully composed with other computational effects, such as failures and exceptions.

\ifcamera
\vspace{0.5cm}
\else
\vspace{0.2cm}
\fi

\setlength\intextsep{-0.5pt}
\begin{wrapfigure}{r}{0\textwidth}
\begin{tabular}{|l|r|}
\hline
Example & LOC\\
\hline
Memory models & 2152\\
Array library & 544\\
File transfer & 630\\
Ariadne & 232\\
\hline
\end{tabular}
\end{wrapfigure}
In \autoref{sec:relatedwork} we discuss related work and
in \autoref{sec:conclusion} we outline future work and conclude.
Additional materials associated with this paper are available from
\url{https://fstar-lang.org/papers/monotonicity}. These include full definitions and
proofs for the formal results in \autoref{sect:metatheory} and
\autoref{sect:reify}, as well as the 
code for all examples in the paper.
The table alongside shows the number of lines of \fstar code for these
examples.
Compared to their code, the listings in the paper are edited for
clarity and sometimes omit uninteresting details.

\section{The Monotonic-State Monad, By Example}
\label{sect:monotonic-state}

Verifying stateful programs within a
dependently typed programming language and proof assistant
is attractive for the expressive power provided by dependent
types, and for the foundational manner in which a program's semantics
can be modeled.
Concretely, we develop our approach in \fstar but our ideas
should transfer to other settings, e.g., Hoare Type
Theory~\citep{nmb08htt}, which \fstar resembles, and also to
other Hoare-style program logics.
Since dependent type theory provides little by way of stateful
programming primitives, we will have to build support for stateful
programming from scratch, starting from a language of pure functions.
We focus primarily on \emph{modeling} stateful programming faithfully,
rather than efficiently \emph{implementing} imperative state. For
instance, we will represent mutable heaps as functions from natural
numbers to values at some type, although an efficient implementation
should use the mutable memory available primitively in
hardware.

We start by defining a simple, canonical state monad, the
basic method by which stateful programming is introduced
in \fstar (\autoref{sec:basic-state}).
We then present the monotonic-state monad, \ls$MST$, in
its simplest, yet most general form: a global state monad
parameterized by a preorder which constrains state
evolution (\autoref{sec:mst}).
We instantiate this generic version of \ls$MST$, first, to
model mutable heaps with allocation, deallocation, untyped and typed
references (\autoref{sec:references}),
and then generalize our heaps further to include per-reference preorders 
(\autoref{sec:mref}).
Throughout, we use small examples to illustrate the pervasive
applicability of monotonic state in many verification scenarios.

\subsection{\fstar: A Brief Review and a Basic State Monad}
\label{sec:basic-state}

We start with a short primer on \fstar and its syntax, showing how to
extend it with a state effect based on a simple state monad.
In particular, \fstar is a programming language with
full dependent types, refinement types, and a user-defined effect
system.
Its effect system includes an inference algorithm based on an
adaptation of Dijkstra's weakest preconditions to higher-order
programs.
Programmers specify precise pre- and post-conditions for their
programs using a notation similar to Hoare Type
Theory and \fstar checks that the inferred
specification is subsumed by the programmer's annotations.
This subsumption check is reduced to a logical validity problem
that \fstar discharges through various means, including a combination
of SMT solving and user-provided lemmas.

\paragraph*{Basic Syntax.} \fstar syntax is
roughly modeled on OCaml
(\ls$val$, \ls$let$, \ls$match$ etc.) although there are many
differences to account for the additional typing features.
Binding occurrences \ls$b$ of variables take the form \ls$x:t$, declaring
a variable \ls$x$ at type \ls$t$; or \ls$#x:t$ indicating that the
binding is for an implicit argument.
The syntax
\ls@fun (b$_1$) ... (b$_n$) -> t@ introduces a lambda abstraction, whereas
\ls@b$_1$ -> ... -> b$_n$ -> c@ is the shape of a curried function type---we
emphasize the lack of enclosing parentheses on
the \ls@b$_i$@. Refinement types are written \ls$b{t}$,
e.g., \ls$x:int{x>=0}$ is the type of non-negative integers
(\ls$nat$).
As usual, a bound variable is in scope to the right of its binding; we
omit the type in a binding when it can be inferred; and for
non-dependent function types, we omit the variable name.
For example, the type of the
pure append function on vectors is written
\ls$#a:Type -> #m:nat -> #n:nat -> vec a m -> vec a n -> vec a (m + n)$, 
with the two explicit arguments and the return type depending 
on the three implicit arguments marked with \ls$#$.

\paragraph*{Basic State Monad and Computation Types.}
Programmers can extend the core \fstar language of pure, total functions
to effectful programs, by providing monadic
representations for the effects concerned. For example, the programmer
can extend \fstar for stateful programming by defining the standard 
state monad (with the type \ls$st a = state -> a * state$)
together with two actions, \ls$get$ and \ls$put$.
%
Given this monad, using a construction described by~\citet{dm4free},
\fstar derives the corresponding \emph{computation type}
\ls$ST t (requires pre) (ensures post)$, describing computations
that, when run in an initial state \ls@s$_0$@ satisfying
\ls@pre s$_0$@, produce a result \ls$r:t$ and a final state \ls@s$_1$@
satisfying \ls@post s$_0$ r s$_1$@.
%
%
On top of that, \fstar also derives a \ls$get$ and a \ls$put$ action for \ls$ST$, with the following types:

\begin{lstlisting}
val get : unit -> ST state (requires (fun _ -> True)) (ensures (fun s$_0$ s s$_1$ -> s$_0$ == s /\ s == s$_1$))
val put : s:state -> ST unit (requires (fun _ -> True)) (ensures (fun _ _ s$_1$ -> s$_1$ == s))
\end{lstlisting}
%
%
%
%
and, when taking
\ls$state = nat$, the \ls$double$ function below has a trivial
precondition (which we often omit) and a
postcondition stating that the final state is twice the initial state.

\begin{lstlisting}
val double : unit -> ST unit (requires (fun _ -> True)) (ensures (fun s$_0$ _ s$_1$ -> s$_1$ == 2 * s$_0$))
let double () = let x = get () in put (x + x)
\end{lstlisting}

\subsection{\ls$MST$: The Monotonic-State Monad}
\label{sec:mst}

\begin{figure}
\begin{lstlisting}
effect MST (a:Type) (requires (pre:(state -> Type))) (ensures (post:(state -> a -> state -> Type)))
(* [get ()]: A standard action to retrieve the current state *)
val get : unit -> MST state (requires (fun _ -> True)) (ensures (fun s$_0$ x s$_1$ -> s$_0$ == x /\ s$_0$ == s$_1$)) $\label{line:MST.get}$
(* [put s]: An action to evolve the state to `s`, when `s` is related to the current state by `rel` *)
val put : s:state -> MST unit (requires (fun s$_0$ -> s$_0$ `rel` s)) (ensures (fun s$_0$ _ s$_1$ -> s$_1$ == s)) $\label{line:MST.put}$
(* [stable rel p]: `p` is stable if it is invariant w.r.t `rel` *)
let stable (#a:Type) (rel:preorder a) (p: (a -> Type)) = forall x y. p x /\ x `rel` y ==> p y $\label{line:MST.stable}$
(* [witnessed p]: `p` was true in some prior program state and will remain true *)
val witnessed : (state -> Type) -> Type $\label{line:MST.witnessed}$
(* [witness p]: A logical action; if `p` is true now, and is stable w.r.t `rel`, then it remains true *)
val witness : p:(state -> Type) -> MST unit (requires (fun s$_0$ -> p s$_0$ /\ stable rel p))$\label{line:MST.witness}$
                                       (ensures (fun s$_0$ _ s$_1$ -> s$_0$ == s$_1$ /\ witnessed p))
(* [recall p]: A logical action; if `p` was witnessed in a prior state, recall that it holds now *)
val recall  : p:(state -> Type) -> MST unit (requires (fun _ -> witnessed p))$\label{line:MST.recall}$
                                     (ensures (fun s$_0$ _ s$_1$ -> s$_0$ == s$_1$ /\ p s$_1$))
(* [witnessed_weaken p q]: `witnessed` is $\mbox{\color{dkgray}{\textit{functorial}}}$ *)
val witnessed_weaken : p:_ -> q:_ -> Lemma ((forall s. p s ==> q s) ==> witnessed p ==> witnessed q)
\end{lstlisting}
  \caption{\ls$MST$: The monotonic-state monad (for \ls$rel:preorder state$)
  }
\label{fig:mst}
\end{figure}


In a nutshell, we use an abstract variant of \ls$st$ parameterized by a
preorder that restricts how the state may be updated---abstraction is
key here, since it allows us to enforce this update condition.
A preorder is simply a reflexive and transitive
relation:

\begin{lstlisting}
let preorder a = rel: (a -> a -> Type) {(forall x . x `rel` x) /\ (forall x y z. x `rel` y /\ y `rel` z ==> x `rel` z)}
\end{lstlisting}
where \ls$ x `rel` y$ is \fstar infix notation for \ls$rel x y$.
(Preorders are convenient, inasmuch as we usually do not wish to track
the actual sequence of state updates.)
Figure~\ref{fig:mst} gives the signature of the abstract
monotonic-state monad \ls$MST$, parameterized by an implicit relation \ls$rel$
of type \ls$preorder state$.
Analogously to \ls$ST$, the \ls$MST$ computation type is also indexed by a
result type \ls$a$, by a precondition on the initial state, and by a
postcondition relating the initial state, the result, and the final
state.
The \ls$bind$ and \ls$return$ for \ls$MST$ are the same as for \ls$ST$.
%
The actions of \ls$MST$ are its main points of interest, although
the \ls$get$ action (line~\ref{line:MST.get}) is still unsurprising---it
simply returns the current state.

\da{There's the slight difference that the subtyping rule for MST 
   ought to provide an assumption that state evolves according to 
   rel, as in the metatheory sections, and as happens when 
   we would apply DM4F to the update monad like representation.}

\paragraph{Enforcing Monotonicity by Restricting \ls$put$.}
The \ls$put$ action requires that the new state be related to the old
state by our preorder (line~\ref{line:MST.put}). This is the main
condition to enable monotonic reasoning.

\paragraph*{Making Use of Monotonicity with \ls$witness$ and \ls$recall$.}
The two remaining actions in \ls$MST$ have no operational
significance---they are erased after verification.
Logically, they capture the intuition that any reachable state evolves
according to \ls$rel$.
The first action, \ls$witness p$ (line~\ref{line:MST.witness}), 
turns a \emph{stable} predicate \ls$p$ valid in the current state
(\IE \ls@p s$_0$@ holds) into a state-independent, logical
capability \ls$witnessed p$.
%
Conversely, the second action, \ls$recall p$, turns such a capability into a property that
holds in the current state \ls@p s$_1$@. In other words, a stable
property, once witnessed to be valid, can be freely assumed to remain
valid (via \ls$recall$) irrespective of any intermediate state
updates, since each of these updates (via the precondition of
\ls$put$) must respect the preorder.

\paragraph*{A Tiny Example: Increasing Counters.} Taking \ls$state = nat$
and \ls$rel = <=$, we can use \ls$MST$ to model an increasing
counter. The \fstar code below (adapted from \S\ref{sect:introduction})
shows how to  use monotonicity to
prove the final assertion, regardless of the stateful \ls$complex_procedure$.
%
\begin{lstlisting}
let x = get () in witness (fun s -> x <= s); complex_procedure (); recall (fun s -> x <= s); assert (x <= get())
\end{lstlisting}

The key point here is that the proposition \ls$witnessed (fun s -> x <= s)$ 
obtained by the \ls$witness$ action is state independent and can thus be
freely transported across large pieces of code, or even unknown ones
like \ls$complex_procedure$. In \fstar{} any state independent
proposition can be freely transported using
either the typing of bind or the rule for subtyping, which plays the
same role as the rule of consequence in classical Hoare Logic.
The following valid Hoare logic rule illustrates this intuition:
\[
\inferrule*[left=(Transporting State-Independent $R$)]
 {\{P\} ~c~ \{Q\}}
 {\{\lambda s.~ P~s \wedge R\} ~c~ \{\lambda s.~ Q~s \wedge R\}}
\]

As such, \ls$MST$ provides a
modular reasoning principle, which is key to scaling verification in the
face of state updates---the success of approaches
like separation logic~\cite{Reynolds2002SL} and dynamic
framing~\cite{Kassios2006} speak to the importance of modular
reasoning. Monotonicity provides another useful, modular 
principle, one that it is quite orthogonal to physical separation---in the example
above, \ls$complex_procedure$ could very well mutate the state
of its context, yet knowing that it only does so in a monotonic manner
allows \ls$x <= get()$ to be preserved across it.
\da{Again, explanation needs improvements as in PCM-based logics 
   one can frame stable properties across function calls.}

\paragraph*{Discussion: Temporarily Escaping the Preorder.}
In some programs, a state update may temporarily break our intended
monotonicity discipline.
For example, consider a mutable 2D point whose coordinates can only
be updated one at a time.
If the given preorder were to require the point to follow some particular
trajectory (e.g., \ls$x=y$), it would prevent any update to \ls$x$
or \ls$y$.

One way to accommodate such examples is to define a more sophisticated preorder
to track states where the original preorder can be temporarily broken.
For instance, if we want the state \ls$s$ to respect a
base \ls@rel$_\textsf{s}$: preorder s@, while temporarily tolerating
violations of this preorder, we could instantiate \ls$MST$
with a richer state type, \ls$t = Ok: s -> t | Tmp: s -> s -> t$,
where `\ls$Tmp snapshot actual$' represents a program state
with `\ls$actual$` as  its current value and `\ls$snapshot$` as its last value
obtained by updates that followed
\ls@rel$_\textsf{s}$@.  We lift \ls@rel$_\textsf{s}$@ to
\ls@rel$_\textsf{t}$@ so that once in \ls$Tmp$, the actual state can
evolve regardless of the preorder \ls@rel$_\textsf{s}$@ (see all the
\ls$Tmp$ cases below):
\begin{lstlisting}
let rel$_\textsf{t}$ t$_0$ t$_1$ = match t$_0$, t$_1$ with
    | Ok s$_0$, Ok s$_1$ | Ok s$_0$, Tmp s$_1$ _ | Tmp s$_0$ _, Ok s$_1$ | Tmp s$_0$ _, Tmp s$_1$ _ -> s$_0$ `rel$_\textsf{s}$` s$_1$
\end{lstlisting}
To temporarily break and then restore the base preorder rel$_\textsf{s}$,
we define the  two actions below,
checking that the current state is related to the last snapshot by \ls@rel$_\textsf{s}$@
before restoring monotonicity.

\begin{lstlisting}
val break : unit -> MST unit (requires (fun t$_0$ -> Ok? t$_0$)) (ensures (fun t$_0$ _ t$_1$ -> let Ok s = t$_0$ in t$_1$ == Tmp s s))
val restore : unit -> MST unit (requires (fun t$_0$ -> match t$_0$ with Ok _ -> False $~$| Tmp s$_0$ s$_1$ -> s$_0$ `rel$_\textsf{s}$` s$_1$))
                           (ensures (fun t$_0$ _ t$_1$ -> let Tmp _ s$_1$ = t$_0$ in t$_1$ == Ok s$_1$))
\end{lstlisting}

\ch{Should explain the ? and ?. syntax, and this is the first place ? appears}\nik{gone ... but probably used later}


\da{[Below:] which style of programming? which corner cases?}

\paragraph*{On the Importance of Abstraction.}
Suppose one were to treat an \ls$MST a (fun _ -> True) (fun _ _ _ -> True)$ computation as a pure
state-passing function of type
\ls@mst a = s$_0$:state -> (x:a * s$_1$:state{s$_0$ `rel` s$_1$})@.
It might seem natural to work with this monadic representation of
state, but it can quickly lead to unsoundness.
To see why, notice that this representation allows the context to pick
an initial state that is not necessarily the consequence of state
updates adhering to the given preorder.
For example, in the code snippet 
below, we first observe that the program state is strictly positive, 
and then define a closure \ls$f$  
that relies on monotonicity to recall that its state, when called, 
is also positive. The precondition of \ls$f$ then requires 
the state-independent proposition 
\ls$witnessed (fun s -> s > 0)$ to be valid, which is provable 
because of the call to \ls$witness (fun s -> s > 0)$ before the 
\ls$let$-binding. More importantly, however, this state-independent 
precondition of \ls$f$ does not put any restrictions on the actual 
state values with which one can call \ls$f ()$, a pure \ls$mst$
state-passing function.
As a result, we can call \ls$f ()$ with whichever state we please, 
e.g.,  \ls$0$, causing a division by zero error at runtime.



\da{I wonder whether it would make sense to informally 
  introduce `reify` here to distinguish `f ()` being defined 
  as an MST computation and it being used as an mst function?}

\ch{Is it intentional that f is defined locally and not at the top
  level? When defined at the top level f has a
  \ls$witnessed (fun s -> s > 0)$ precondition that we could try to
  use to infer the pre-condition of its reification. Does f still have
  the same precondition when defined locally? Nope, it can be weakened
  out by subtyping, which seems like a bug not a feature!}

\begin{lstlisting}
put(get() + 1); witness (fun s -> s > 0); let f () = recall (fun s -> s > 0); 1 / get() in f () 0 (* $\longleftarrow\mbox{\color{dkred}\textbf{BROKEN!}}$ *)
\end{lstlisting}


For the next several sections, we treat \ls$MST$ abstractly. We return
to this issue in \autoref{sect:reify} and carefully introduce two
coercions, \ls$reify$ and \ls$reflect$, to turn \ls$MST$ computations
into \ls$mst$ functions and back, while not compromising
monotonicity. Using \ls$reify$, we show how
to prove relational properties of \ls$MST$ computations, e.g.,
noninterference; and using \ls$reflect$, we show how to add
new actions that respect the preorder
(while potentially temporarily violating it in an unobservable way).

\subsection{Heaps and References, Both Untyped and Typed}
\label{sec:references}

The global monotonic state provided by \ls$MST$ is a useful primitive,
but for practical stateful programming we need richer forms of
state. In this section we show how to instantiate the
state and preorder of \ls$MST$ to model references to
mutable, heap-allocated data. Our references come in two varieties.
Untyped references (type \ls$uref$) represent transient locations in
the heap whose type may change as the program evolves, until they are
explicitly deallocated.
Typed references (type \ls$ref t$) are always live (at least observably so, given
garbage collection) and contain a \ls$t$-typed value.

Prior works on abstractly encoding mutable heaps within dependent
types, e.g., HTT~\citep{nmb08htt} and CFML~\citep{Chargueraud2011CF},
only include untyped references; while primitive and general, they are
also harder to use safely, since one must maintain explicit liveness
invariants of the form $(r \mapsto_\tau \_)$ stating that an
untyped reference contains a value of type $\tau$ in the current
state. \citepos{swierstra09thesis} shape-indexed references are more
sophisticated, but still require a proof to accompany each use of the
reference to establish that the current heap contains it.
Using our monotonic-state monad \ls$MST$, we show how to directly
account for both typed and untyped references, where typed references
are just as easy to use as in more mainstream, ML-like languages,
without the need for liveness invariants or explicit proof
accompanying each use: a value
\ls$r:ref t$ is a proof of its own well-typed, membership in the current heap.

%
%

We model \ls$heap$ (see code below) as a map from abstract {\em
identifiers}
%
to either \ls$Unused$ or \ls$Used$ cells, together with a
counter (\ls$ctr$) tracking the next \ls$Unused$ identifier.
As in the CompCert memory model \cite{LeroyB08}, identifiers are
internally represented by natural numbers (type \ls$nat$) but these
numbers are not connected to the real addresses used by an efficient,
low-level implementation.
New cell identifiers are freshly generated by bumping the counter and are
never reused.
%
%
\begin{lstlisting}
type tag = Typed : tag | Untyped : bool -> tag
type cell = Unused | Used : a:Type -> v:a -> tag -> cell
type heap = H : h:(nat -> cell) -> ctr:nat{forall (n:nat{ctr <= n}). h n == Unused} -> heap
\end{lstlisting}
A `\ls$Used a v tag$' cell is a triple, where \ls$a$ is a type
and \ls$v$ is a value of type \ls$a$. As such, heaps are
heterogeneous maps (potentially) storing values of different types for
each identifier.
The \ls$tag$ is either \ls$Typed$, marking a cell referred to by a
typed reference; \ls$Untyped true$, for a live, allocated untyped cell;
or \ls$Untyped false$, for a cell that was once live
but has since been deallocated. (We distinguish
\ls$Untyped false$ from \ls$Unused$ to simplify our model of freshness---a
client of our library can treat a newly allocated reference as being
distinct from all previously allocated ones.)


Using monotonicity, we now show how to define our two kinds of references.
In particular, a \emph{typed reference} \ls$ref t$ is an identifier \ls$id$ for which the heap has
been witnessed to contain a \ls$Used$, \ls$Typed$ cell of type \ls$t$.
An \emph{untyped reference} \ls$uref$, on the other hand, has a much weaker invariant: it was
witnessed to contain a \ls$Used$, \ls$Untyped$ cell that could have
since been deallocated.

\begin{lstlisting}
let has_a_t (id:nat) (t:Type) (H h _) = match h id with Used a _ Typed -> a == t | _ -> False
abstract type ref t = id:nat{witnessed (has_a_t id t)}
let has (id:nat) (H h _) = match h id with Used _ _ (Untyped _) -> True $~$| _ -> False
abstract type uref  = id:nat{witnessed (has id)}
\end{lstlisting}

To enforce these invariants on state-manipulating operations,
we define a preorder \ls$rel$ on \ls$heap$ that
constrains the heap evolution. It states that
every \ls$Used$ identifier remains \ls$Used$;
every \ls$Typed$ reference has a stable type; and
that no \ls$Untyped$ reference may be reused after deallocation.

\begin{lstlisting}
let rel (H h$_0$ _) (H h$_1$ _) = forall id. match h$_0$ id, h$_1$ id with
                            | Used a _ Typed, Used b _ Typed -> a == b
                            | Used _ _ (Untyped live$_0$), Used _ _ (Untyped live$_1$) -> not live$_0$ ==> not live$_1$
                            | _ -> False
\end{lstlisting}


Instantiating \ls$MST$ with \ls$state=heap$ and the preorder
above, we can implement the expected operations for allocating,
reading, writing typed references.
The \ls$alloc$ action below allocates a
new, typed reference \ls$ref t$ by generating a fresh identifier \ls$id$;
extending the heap at \ls$id$ with a new typed cell; and witnessing the
new state, ensuring that \ls$id$ will contain
a \ls$t$-typed cell for ever. Reading and writing a reference are similar: they
both \ls$recall$ the reference exists in the heap at its expected
type.

\begin{lstlisting}
let alloc #t (v:t) : MST (ref t) (ensures (fun h$_0$ id h$_1$ -> fresh id h$_0$ h$_1$ /\ modifies {} h$_0$ h$_1$ /\ h$_1$.[id] == v)) =
    let H h id = get () in put (H (upd h id (Used t v Typed)) (id + 1)); witness (has_a_t id t); id
let (!) #t (r:ref t) : MST t (ensures (fun h$_0$ v h$_1$ -> h$_0$ == h$_1$ /\ has_ref r h$_1$ /\  h$_1$.[r] == v)) =
    recall (has_ref r); let h = get () in h.[r]
let (:=) #t (r:ref t) (v:t) : MST unit (ensures (fun h$_0$ _ h$_1$ -> modifies {r} h$_0$ h$_1$ /\ has_ref r h$_1$ /\ h$_1$.[r] == v)) =
    recall (has_ref r); let H h ctr = get () in put (H (upd h r (Used t v Typed), ctr))
\end{lstlisting}
\cf{Hard to parse; let's expand it if we have some room left.}

These functions make use of a few straightforward auxiliary
definitions for freshness of identifiers (\ls$fresh$) and for the
write-footprint of a computation (\ls$modifies$). The one subtlety is
in the definition of \ls$h.[r]$, a total function to select a
reference \ls$r$ from a heap \ls$h$. Unlike the stateful (\ls$!$)
operator, \ls$h.[r]$ has a precondition requiring that \ls$h$ actually
contain \ls$r$---even though the type of \ls$r$ indicates that it has been
witnessed in some prior heap of the program, this does not
suffice to recall that it is actually present in an arbitrary
heap \ls$h$. In other words, pure functions may not
use \ls$recall$. On the other hand, the stateful lookup (\ls$!$)
is free to \ls$recall$ the membership of \ls$r$ in the current
heap in order to meet the precondition of \ls$h.[r]$.%
\footnote{In our revision to the libraries of \fstar,
  we use a more sophisticated representation of \ls$ref t$ that enables
  a variant of \ls$h.[r]$ without the \ls$has_ref h r$ precondition. This variant is convenient to use in specifications,
since its well-typedness is easier to establish.
We omit it for simplicity, but the curious reader may
consult the \ls$FStar.Heap$ library for the full story.
\nik{revised this}\ch{Why would one care to remove the
precondition if above it wasn't a problem? I guess it's because it
makes the specification style heavier, in which case should we mention
it here?}  \da{Correct me if I am wrong, but we can only remove the
precondition in ghost code? In order to actually define \ls$!$ above
we need a tot version of \ls$h.[r]$.\nik{right}}}

\begin{lstlisting}
let has_ref #t (r:ref t) h = has_a_t r t h
let fresh #t (r:ref t) (H h$_0$ _) (H h$_1$ _) = h$_0$ r == Unused /\ has_ref r h$_1$
let modifies (ids:set nat) (H h$_0$ _) (H h$_1$ _) = forall id. id $\not\in$ ids /\ Used? (h$_0$ id) ==> h$_0$ id == h$_1$ id
let `_.[_]` #t h (r:ref t{has_ref r h}) : t = let H h _ = h in match h r with Used _ v _ -> v
\end{lstlisting}

The operations of untyped references are essentially simpler
counterparts of \ls$alloc$, \ls$(!)$ and \ls$(:=)$ with weaker types.
The \ls$free$ operation is easily defined by replacing
an \ls$Untyped$ cell with a version marking it as deallocated. The
precondition of \ls$free$ prevents double-frees, and is necessary to
show that the preorder is preserved as we mark the cell deallocated.

\begin{lstlisting}
let live (r:uref) (H h _) = match h r with Used _ _ (Untyped live) -> live | _ -> false
let free (r:uref) : MST unit (requires (live r)) (ensures (fun h$_0$ _ h$_1$ -> modifies {r} h$_0$ h$_1$)) =
   let H h ctr = get () in put (H (upd h r (Used unit () (Untyped false))) ctr)
\end{lstlisting}

\subsection{Monotonic References}
\label{sec:mref}

Typed references \ls$ref t$ use a fixed global preorder saying
that the type of each ref-cell is invariant.
However, we would like a more flexible form, allowing the programmer to associate
a preorder of their choosing with each typed reference.
In this section, we present a library for providing a type
`\ls$mref a ra$', a typed reference to a value of type `\ls$a$' whose
contents is constrained to evolve according to a preorder `\ls$ra$' on `\ls$a$'.
Using \ls$mref$s, one can for instance encode a form of typestate
programming~\citep{SY86Typestate} by attaching to an \ls$mref$
a preorder that corresponds to the reachability relation of a state machine.

As above, our global, monotonic-state monad \ls$MST$ can be
instantiated with a suitable \ls$heap$ type for the global state (defined below)
and a preorder on the global state that is intuitively the pointwise
composition of the preorders associated with
each \ls$mref$ that the state contains. In this setting, the
type \ls$ref t$ can be reconstructed as a derived form, i.e.,
\ls$ref t = mref t (fun _ _ -> True)$



\paragraph{An Interface for Monotonic References.}
When allocating a monotonic reference one picks both the initial value
and the preorder constraining its evolution. An \ls$mref a ra$ can be
dereferenced unconditionally, whereas assigning to an \ls$mref a ra$
requires maintaining the preorder, analogous to the precondition
on \ls$put$ for the global state.

\begin{lstlisting}
type mref : a:Type -> ra:preorder a -> Type
val (:=) : #a:Type -> #ra:preorder a -> r:mref a ra -> v:a -> MST unit
    (requires (fun h -> h.[r] `ra` v))
    (ensures  (fun h$_0$ _ h$_1$ -> modifies {r} h$_0$ h$_1$ /\ h$_1$.[r] == v))
\end{lstlisting}

\noindent The local state analog of \ls$witness$ on the global state
allows observing a predicate \ls$p$ on the global \ls$heap$ as long as
the predicate is \ls$stable$ with respect to arbitrary heap updates
that respect the preorder only on a given reference. Using \ls$recall$
to restore a previously witnessed property remains unchanged.
\begin{lstlisting}
let stable #a #ra (r:mref a ra) (p:(heap -> Type)) = forall h$_0$ h$_1$. p h$_0$ /\ h$_0$.[r] `ra` h$_1$.[r] ==> p h$_1$
val witness : #a:Type -> #ra:preorder a -> r:mref a ra -> p:(heap -> Type){stable r p} -> MST unit
    (requires (fun h -> p h))
    (ensures  (fun h$_0$ v h$_1$ -> h$_0$==h$_1$ /\ witnessed p))
\end{lstlisting}

\paragraph{Implementing Monotonic References.} To implement \ls$mref$ we
choose the following, revised representation of \ls$heap$ and its
global preorder. We enrich the tags from \S\ref{sec:references} to
additionally record a preorder with every typed cell. Correspondingly,
the global preorder on heaps is, as mentioned earlier, the pointwise
composition of preorders on typed cells
(the \ls$cell$ and \ls$heap$ types are unchanged).

\begin{lstlisting}
type tag a = Typed : preorder a -> tag a | Untyped : bool -> tag a
type cell = Unused | Used : a:Type -> a -> tag a -> cell
type heap = H : h:(nat -> cell) -> ctr:nat{forall (n:nat{ctr <= n}). h n == Unused} -> heap
let rel (H h$_0$ _) (H h$_1$ _) = forall id. match h$_0$ id, h$_1$ id with
   | Used a$_0$ v$_0$ (Typed ra$_0$), Used a$_1$ v$_1$ (Typed ra$_1$) -> a$_0$ == a$_1$ /\ ra$_0$ == ra$_1$ /\ v$_0$ `ra$_0$` v$_1$
   | Used _ _ _ (Untyped live$_0$), Used _ _ _ (Untyped live$_1$) -> not live$_0$ ==> not live$_1$
   | _ -> False
let mref a ra = id:nat{witnessed (fun (H h _) -> match h id with Used b _ (Typed rb) -> a == b /\ ra == rb)}
\end{lstlisting}

Monotonic references are our main building block for defining more
sophisticated abstractions.
While we have shown them here in the context of a flat heap, our
libraries provide monotonic references
within \emph{hyper-heaps}~\citep{mumon}
and \emph{hyper-stacks}~\citep{lowstar}, more sophisticated,
region-based memory models used in \fstar to keep track of
object-lifetimes and to encode a weak form of separation. \ch{TODO this is
  an overstatement given the current status of our artifact!}
%


\section{Monotonic state in action: A secure file-transfer application}
\label{sec:action}

Consider transferring a file $f$ from a sender application $S$ to a
receiver application $R$. To ensure that the file is transferred
securely, the applications rely on a protocol $P$, which is designed
to ensure that $R$ receives exactly the file that $S$ sent (or detects
a transmission error) while no one else learns anything about $f$. For
instance, $P$ could be based on TLS, and use cryptography and
networking to achieve its goals, but its details are unimportant. Our
example addresses the following concerns:

\begin{enumerate}[leftmargin=0.55cm]
\item \emph{Low-level buffer manipulation}: For efficiency
reasons, $S$ and $R$ prepare byte buffers shared with~$P$. For
instance, the receiver $R$ allocates an uninitialized buffer and
requests $P$ to fill the buffer with the bytes it receives.
Monotonicity ensures that once memory is initialized
it remains so.\ch{Should say something about freezing too}

\cf{TODO: explain fragmentation and authenticity first? motivate
  low-level explicit memory management, meant to prevent multiple
  buffering.  Bug we would detect: authenticated chunks of the file
  overwritten before they are processed by $R$.}

\item \emph{Modeling distributed state}: To state and prove
the correctness of a distributed system, even one as simple as our
2-party file-transfer scenario, it is common to describe the state of
the system in terms of some global \emph{ghost} (i.e., purely
specificational) state. We use monotonicity to structure this ghost
state as an append-only log of messages sent so far, ensuring that an
observation of the state of the protocol remains consistent even while
the state evolves.

\item \emph{Fragmentation and authenticity}: Protocol $P$
dictates a maximum size of messages that can be sent at once. As
such, the sender has to fragment the file $f$ into several
chunks and the receiver must reconstruct the file. We use
monotonicity to show that $R$ always reads a prefix of the
stream that $S$ has sent so far, i.e., $R$ receives authentic file
fragments from $S$ in the correct order.

\item \emph{Secrecy}: Finally, we consider possible implementations
of the protocol $P$ itself and prove that under certain,
standard cryptographic hypotheses, $P$ leaks no information about the
file $f$, other than some information about its length.
\end{enumerate}

\subsection{Safely Using Uninitialized and Frozen Memory}
\label{sec:uninitialized}

Reasoning about safety in the presence of uninitialized memory is a
well-known problem, with many bespoke program analyses targeting it,
e.g.,~\citepos{qi09masked} masked types.
The essence of the problem involves reasoning about monotonic
state---memory is unreadable until it is initialized, and remains
\emph{readable} thereafter.
A dual problem is deeming an object no longer \emph{writable}. For
instance, after validating its contents, one may want to freeze an
object in a high-integrity state.
\cf{An interesting sub-case, applicable to file-transfer, is that the
  content is frozen once initialized. Where to discuss it? Also
  discuss gradual initialization? range-based initialization?}

Using monotonic references, we designed and implemented a verified
library for modeling the safe use of uninitialized arrays that may
eventually be frozen, including support for a limited form of pointer
arithmetic to refer to prefixes or suffixes of the array.\ch{If you
  can take both prefixes and suffixes you can construct arbitrary
  subarrays, right? If so maybe phrase it that way?}
We sketch a fragment of this library here, showing only the parts
relevant to the treatment of uninitialized memory.

The main type provided by our library is an abstract type
`\ls$array a n$' for a possibly uninitialized array indexed by its
contents type \ls{a} and length \ls{n}. An \ls{array} is implemented
under the hood by using a monotonic reference containing a \ls$seq (option a)$ and constrained by the preorder \ls{remains_init}.
\begin{lstlisting}
abstract type array (a:Type) (n:nat) = mref (repr a n) remains_init
$\mbox{\textit{where}}$ repr a n = s:seq (option a){len s == n}
$\mbox{\textit{and\;\;\;}}$ remains_init #a #n (s$_0$:repr a n) (s$_1$:repr a n) = forall (i:nat{i < n}). Some? s$_0$.(i) ==> Some? s$_1$.(i)
\end{lstlisting}
\noindent Notice the interplay between refinements types
and monotonic references. The refinement type
\ls$s:seq (option a){len s == n}$ passed to \ls$mref$
constrains the stored sequence to be of the appropriate length in
every state.
Though concise and powerful, refinements of this form can only
enforce invariants of each reachable program state {\em taken in
isolation}.
In order to constrain how the array contents can {\em evolve},
the preorder \ls$remains_init$ states that the sequence
underlying an array can evolve from \ls{s$_0$} to \ls{s$_1$} only if every
initialized index in \ls{s$_0$} remains initialized in \ls{s$_1$}.

Given this representation of \ls$array$, the rest of the
code is mostly determined.
For instance, its \ls$create$ function takes a length \ls$n$ but no initial
value for the array contents.
\begin{lstlisting}
abstract let create (a:Type) (n:nat) : ST (array a n) (ensures (fun h$_0$ x h$_1$ -> fresh x h$_0$ h$_1$ /\ modifies { } h$_0$ h$_1$)) =
  alloc (Seq.create n None) remains_init
\end{lstlisting}

The pure function \ls$as_seq h x$ enables reasoning about an
array \ls$x$ in state \ls$h$ as a sequence of optional values.
Below we use it to define which parts of an array are initialized, i.e.,
those indices at which the array contains \ls$Some$ value.
\begin{lstlisting}
let index #a #n (x:array a n) = i:nat{i < n}
let initialized #a #n (x:array a n) (i:index x) (h:heap) = Some? (as_seq h x).(i)
\end{lstlisting}

\noindent The main use of monotonicity in our library is to
observe that \ls$initialized$ is stable with respect to \ls$remains_init$, the
preorder associated with an array. As such, we can
define a state-independent proposition \ls@(x `init_at`$ $ i)@ using
the logical \ls$witnessed$ capabilities.
\begin{lstlisting}
let init_at #a #n (x:array a n) (i:index x) = witnessed (initialized x i)
\end{lstlisting}

We can now prove that writing to an array at index \ls$i$ ensures that it
becomes initialized at \ls$i$, which is a necessary precondition to read
from \ls{i}. Notice the use of \ls$witness$ when writing, to
record the fact that the index is initialized and will remain so; and
the use of \ls$recall$ when reading, to recover that the array is
initialized at \ls$i$ and so the underlying sequence contains
\ls$Some v$ at \ls$i$.
%
\begin{lstlisting}
abstract let write (#a:Type) (#n:nat) (x:array a n) (i:index x) (v:a) : ST unit
 (ensures  (fun h$_0$ _ h$_1$ -> modifies {x} h$_0$ h$_1$ /\ (as_seq h$_1$ x).(i) == Some v /\ x `init_at` i)) =
 x := Seq.upd !x i (Some v); witness (initialized x i)
 
abstract let read (#a:Type) (#n:nat) (x:array a n) (i:index x{x `init_at` i}) : ST a
 (ensures (fun h$_0$ r h$_1$ -> modifies { } h$_0$ h$_1$ /\ Some r == (as_seq h$_0$ x).(i))) =
 recall (initialized x i); match !r.(i) with Some v -> v
\end{lstlisting}

It is worth noting that in the absence of monotonicity, 
the \ls$init_at$ predicate defined above would need
to be state-dependent, and thus carried through the subsequent stateful
functions as a stateful invariant, causing unnecessary additional 
proof obligations and bloated specifications.


\paragraph*{Freezing and Sub-Arrays.}
\setlength\intextsep{-0.5pt}
\begin{wrapfigure}{r}{0\textwidth}
  \includegraphics[width=0.25\textwidth]{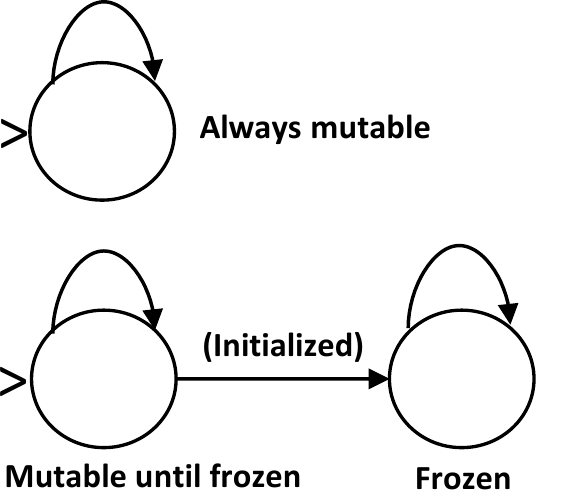}
\end{wrapfigure}
Our full array library supports freezing an array once it is fully
initialized. We add a stateful predicate \ls{is_mutable x h}
indicating that the array \ls{x} is mutable in state \ls{h}.
By design, \ls{is_mutable} is not a stable
predicate--freezing an array explicitly revokes
mutability. While \ls{is_mutable} is a precondition of \ls{write},
freezing an initialized array provides a stable
predicate \ls{frozen_with x s}, where \ls{s:seq a} represents the
stable snapshot of the array contents--the clients can later recall
that \ls{x}'s contents still correspond to \ls{s}.
The library also supports \emph{always mutable} arrays in the same
framework, for which \ls{is_mutable} is indeed a stable predicate. Consequently,
the \ls$write$ operation on always mutable arrays has no stateful
precondition, since its argument's mutability can just be recalled.
Internally, these temporal properties of an array are managed by a
preorder capturing the state-transition system shown alongside.
%
%
Aside from these core operations, our library provides
functions to create aliases to a \ls$prefix$ or a \ls$suffix$ of an
array, while propagating information about the fragments of the array
that are initialized or frozen to the aliases.

\subsection{Modeling the Protocol's Distributed State}
\label{sec:protocol_iface}

Beyond simple safety, to prove our file-transfer application correct
and secure, we first need to model the protocol $P$.
Conceptually, we model each instance of the protocol using a \emph{ghost}
state shared between the protocol participants (in
this case, just $S$ and $R$).
The shared state contains a log of message fragments already sent by
$S$; the main invariant of the protocol is that
successful \ls$receive$ operations return fragments from a prefix of
the stream of fragments sent so far.
This style of modeling distributed systems has a long
tradition~\citep{ChandyLamport85}, and several recent program logics
and verification systems have incorporated special support for such
shared ghost state~\citep{f7,fstar-jfp13,SergeyWT18}---we just make
use of our monotonic references library for this.
\cf{Distributed algorithms are all about stable properties---one
  usually cannot rescind a message.  We may also mention this style
  for crypto modelling, in relation to prior work on \fstar e.g. at
  Oakland'17.}

\cf{State upfront what's the intended goal of the protocol: is it
  about communicating a stream of bytes, a stream of fragments, or a
  set of fragments? (Not sure which is best for file transfer,
  arguably AEAD fragments on a lossy network.) Same question in its
  implementation.}

\paragraph{Connection State.}
The state of a protocol instance
is the abstract type \ls$connection$. Its interface (shown below)
provides a function \ls$log c h$, representing the messages sent so
far on \ls$c$; and a sequence counter \ls$ctr c h$, representing the
current position in the log.\ch{Didn't understand why these things
  take heap arguments. What is it that one needs to pass there?}
Whereas the sender and receiver each maintain their own counters,
the \ls$log$ is shared, specification-only state between the two
participants.\footnote{The \ls$log$ is only needed for specification
and modeling purposes. So, in practice, we maintain the log in
computationally irrelevant state, which \fstar supports. However, we
elide this level of detail from our example here, since it is orthogonal.}

\begin{tabular}{lcr}
\begin{lstlisting}
type connection
type message = s:seq byte{len s <= fragment_size}

val log: connection -> heap -> seq message
val ctr: connection -> heap -> nat
val is_receiver: connection -> bool

let receiver = c:connection{is_receiver c}
let sender = c:connection{not (is_receiver c)}
\end{lstlisting}
&
$\qquad$
&
\raisebox{-50pt}{\includegraphics[width=0.35\textwidth]{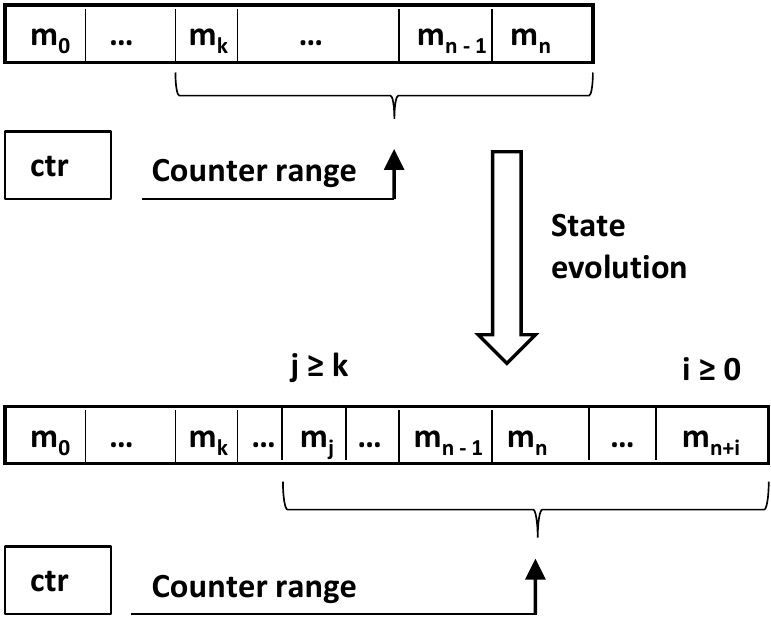}}
\end{tabular}

\paragraph{Monotonic Properties of the Protocol.}
Monotonicity comes into play with the \ls$snap$ operation. A snapshot
of the distributed protocol state remains stable as the state
evolves. Obtaining such stable snapshots is a basic component of
designing and verifying distributed protocols, and our libraries make
this easy to express. In particular, the snapshotted log remains a
prefix of the log and the counter never decreases. The clients can
also recall that in any state, the counter value is bounded by the
length of the log, i.e., the valid transitions of the counter are
dependent on the log, as depicted in the figure above.

\begin{lstlisting}
let snapshot c h$_0$ h = log c h$_0$ `is_a_prefix_of` log c h /\ ctr c h$_0$ <= ctr c h /\ ctr c h$_0$ <= len (log c h)
val snap: c:connection -> MST unit (ensures (fun h$_0$ _ h$_1$ -> h$_0$ == h$_1$ /\ witnessed (snapshot c h$_0$)))
val recall_counter: c:connection -> MST unit (ensures (fun h$_0$ _ h$_1$ -> h$_0$ == h$_1$ /\ ctr c h$_0$ <= len (log c h$_0$)))
\end{lstlisting}

\paragraph{Receiving a Message.} Receiving a message using
\ls$receive buf c$ requires \ls$buf$ to contain enough space
for a message, and that \ls{buf} and \ls{c} use disjoint state. It
ensures that only the connection and the buffer
are modified; that at most \ls{fragment_size} bytes are received
into \ls$buf$; and the bytes correspond to the message recorded in
the log at the current counter position, which is then incremented.

\begin{lstlisting}
val receive: #n:nat{fragment_size <= n} -> buf:array byte n -> c:receiver{disjoint c buf} -> MST (option nat) $\label{line:protocol.receive}$
 (ensures  (fun h$_0$ ropt h$_1$ -> match ropt with None -> h$_0$ == h$_1$
                        | Some r -> modifies {c, buf} h$_0$ h$_1$ /\ r <= fragment_size /\ modifies_array buf 0 r /\
                                  all_init (prefix buf r) /\ ctr c h$_1$ == ctr c h$_0$ + 1 /\ log c h$_1$ == log c h$_0$ /\
                                  sub_sequence (as_Seq buf h$_1$) 0 r == (log c h$_0$).(ctr c h$_0$)))
\end{lstlisting}

\paragraph{Sending a Message.} Sending a message using \ls$send buf c$,
requires \ls$buf$ to be an array of initialized bytes.\footnote{Of
course, \ls$send$ and \ls$receive$ actually send or receive messages
on the network, so they have more than just \ls$MST$ effect. However,
for simplicity here, we model IO in terms of state, as detailed
in \S\ref{sec:protocol_impl}.}
The postcondition ensures that only the connection state is modified,
at most \ls{fragment_size} bytes are sent, the \ls{ctr} is incremented
by one, and the \ls{log} is extended with the sent message.

\begin{lstlisting}
val send : #n:nat -> buf:array byte n{all_init buf} -> c:sender{disjoint c buf} -> MST nat
  (ensures (fun h$_0$ sent h$_1$ -> modifies {c} h$_0$ h$_1$ /\ sent <= min n fragment_size /\ ctr c h$_1$ == ctr c h$_0$ + 1 /\
                         log c h$_1$ == snoc (log c h$_0$) (sub_sequence (as_seq buf h$_0$) 0 sent)))
\end{lstlisting}


\paragraph*{Look Ma, No Stateful Invariants!} 
To reiterate the importance of monotonicity, note that a salient feature of
our protocol interface is its lack of stateful preconditions and
invariants. This means, for instance, that we may create and
interleave the use of several connections without needing to worry
about interference among the instances. In contrast, were our
interface to use stateful invariants (e.g., that the counter is always
bounded by the length of the log, among other properties), one would
need to prove that the invariant is maintained through, or framed
across, all state updates. Of course, one should not expect to always
eliminate stateful invariants; however, monotonic state, when
applicable, can greatly simplify the overhead of maintaining stateful
invariants.

\subsection{Implementing the Protocol Securely}
\label{sec:protocol_impl}

There are many possible ways to securely implement our protocol
interface. We chose perhaps the simplest option, assuming the sender
and the receiver share a private source of randomness, and use one-time pads
for perfectly secure encryption with a message authentication code
(MAC) authenticating the cipher and a sequence number.
However, more complex and broadly deployable alternatives have been
proven secure using monotonic state in \fstar, including the main
``authenticated encryption with additional data'' (AEAD)
constructions used in TLS 1.3~\citep{record}.

\ch{Should make it explicit that this provides an implementation for
  the abstract connection type from prev section}

\paragraph{Connection.} A \ls$connection$ comes in two varieties: a pair
\ls$(S rand entries)$ for the sender, or a triple \ls$(R rand entries ctr)$
for the receiver. Both sides share the same source of randomness (a
stream of random fragments) and a log of entries; the receiver
additionally has a monotonic counter to keep track of where in the
stream it is currently reading from.
Each entry in the log is itself a $4$-tuple of: an index \ls$i$ into the
randomness;
the plain text \ls$m$;
the cipher text \ls$c$, with an invariant claiming it to be computed
from \ls$m$ via the one-time pad;
and the message authentication code, \ls$mac$.
We supplement the log with an invariant that the $i$th entry in the
log has index $i$---this will be needed for proving the receiver correct.

\begin{lstlisting}
type network_message = s:seq byte{len s == fragment_size}
type randomness = nat -> network_message
type entry (rand:randomness) =
  | Entry : i:nat -> m:message -> c:network_message{pad m $\oplus$ rand i == c} -> mac:tag i c -> entry
type entries rand = s:seq (entry rand){forall (i:nat{i < len s}). let Entry j _ _ _ = s.(i) in j == i}
type connection =
  | S : rand:randomness -> entries:mref (entries rand) is_a_prefix_of -> connection
  | R : rand:randomness -> entries:mref (entries rand) is_a_prefix_of
      -> ctr:mref (n:nat{witnessed (fun h -> n <= len h.[entries])}) increasing -> connection
\end{lstlisting}

\paragraph{Logs, Counters, and Snapshots.}
The types above show several forms of interplay between monotonicity,
stable predicates, and refinement types. The
preorders \ls$`is_a_prefix_of`$ and \ls$`increasing`$ restrict how the log
and the counter evolve.
Within an \ls$Entry$, the refinement on the cipher enforces an
invariant on the data stored in the log.
Perhaps most interestingly, the type of the counter mixes a refinement
and a stable predicate: it states that the monotonic reference holding
the counter contains \ls$n:nat$ that is guaranteed to be bounded by
the number of entries in the current state. This combination encodes a
form of monotonic dependence among multiple mutable locations, allowing
the preorder of the counter to evolve as the log evolves, i.e., a log
update allows the counter to be advanced further along its increasing
preorder.

With our type definitions in place, we can implement the signature
from the previous subsection.

\begin{lstlisting}
let entries_of (S _ es | R _ es _) = es
let log c h = Seq.map (fun (Entry _ msg _ _) -> msg) h.[entries_of c] (* the plain texts sent so far *)
let ctr c h = match c with S _ es -> len h.[es] | R _ _ ctr -> h.[ctr] (* write at end; read from a prefix *)
let recall_counter c = match c with S _ _ -> () | R _ es ctr -> let n = !ctr in recall (fun h -> n <= len h.[es])
let snap c = let h$_0$ = ST.get () in recall_counter c; witness (snapshot c h$_0$)
\end{lstlisting}

\paragraph*{Encrypting and Sending a Message Fragment.} To implement
\ls$send buf c$, we take a message of size at
most `\ls{min n fragment_size}' from \ls$buf$, pad it if needed, encrypt it using
the current key from \ls$rand$, and compute the MAC using the
cipher and the current counter.
We then add a new entry to the log, preparing to send the
cipher and the MAC, and then call \ls$network_send c cipher mac$, which
expects the \ls$cipher$ and \ls{mac} to be the same as in the
last log entry.

\paragraph*{Receiving, Authenticating and Decrypting a Message.}
The receiver's counter indicates the sequence number $i$ of the next
expected fragment. To receive it, we parse the bytes received from the
network into a pair of a cipher and a MAC. We then authenticate
that this is indeed the $i$th message sent using \ls$mac_verify$,
which, based on a cryptographic model of the MAC, guarantees
that the log contains an appropriate entry for the cipher, sequence
number, and MAC.
We then decipher the message using the one-time pad, and the invariant
on the entries (together with properties about \ls$pad$, \ls$unpad$,
and $\oplus$) guarantees that the received message \ls$msg$ is
indeed recorded in the log of entries. Before returning the number of
bytes received, we must increment the counter. Manipulating the
counter's combination of refinement and stable predicate requires a bit
of care: we have to first witness that the new value of the counter
will remain bounded by length of the log.
\begin{lstlisting}
let receive #n buf c = let R rand entries ctr = c in let i = !ctr in
  match parse (network_receive c) with
  | None -> None
  | Some (cipher, mac) ->
    if mac_verify cipher mac i entries then (* guarantees that (Entry i _ cipher mac _) is in entries *)
      let msg = unpad (cipher $\oplus$ rand i) in Array.fill buf msg;
      witness (fun h -> i + 1 <= len h.[entries]);  ctr := i + 1; Some (len msg)
    else None
\end{lstlisting}

\subsection{Correctness, Authenticity, and Secrecy of File Transfer}
\label{sec:toplevel}

Using this protocol, we program and verify the top-level
applications for sending and receiving entire files. Our goal is to
show that the sender application $S$
successfully uses the protocol to fragment and send a file and that
the receiver reconstructs exactly that file, i.e., file transfer is
correct and authentic. Additionally, we prove that file transfer is
confidential, i.e., under a suitable cryptographic model, a
network adversary gains no information about the transferred file.

We first specify what it means to send a file correctly in terms of
the abstractions provided by the underlying protocol. We
use `\ls$sent_bytes f c from to$' to indicate that the protocol's log
contains exactly the contents of the file \ls$f$ starting at
position \ls$from$ until \ls{to}. The stable property \ls{sent f c}
then expresses that the file \ls$f$ was sent on connection \ls$c$ at some point
in the past.

\begin{lstlisting}
let sent_bytes (file:seq byte) (c:connection) (from:nat) (to:nat{from <= to}) (h:heap) =
  let log = log c h in to <= len log /\ file == flatten (sub_sequence log from to)
let sent (file:seq byte) (c:connection) = $\exists$ from to. witnessed (sent_bytes file from to)
\end{lstlisting}

The sender's top-level function calls into the protocol repeatedly
sending message-sized chunks of the file, until no more bytes
remain. Its specification ensures the file was indeed sent.
\begin{lstlisting}
let send_file (#n:nat) (file:array byte n{all_init file}) (c:sender{disjoint c file}) : MST unit
  (ensures  (fun h$_0$ _ h$_1$ -> modifies {c} h$_0$ h$_1$ /\ sent (as_seq file h$_0$) c))
  = let rec aux (from:nat) (pos:nat{pos <= n}) : MST unit
     (requires (fun h$_0$ -> from <= ctr c h$_0$ /\ sent_bytes (sub_seq (as_seq file h$_0$) 0 pos) c from (ctr c h$_0$) h$_0$))
     (ensures  (fun h$_0$ _ h$_1$ -> modifies {c} h$_0$ h$_1$ /\ from <= ctr c h$_1$ /\
                           sent_bytes (as_seq file h$_0$) c from (ctr c h$_1$) h$_1$))
      = if pos <> n then let sub_file = suffix file pos in let sent = send sub_file c in aux from (pos + sent) in
    let h$_0$ = ST.get () in aux (ctr c h$_0$) 0;
    let h$_1$ = ST.get () in witness (sent_bytes (as_seq file h$_0$) c (ctr c h$_0$) (ctr c h$_1$))
\end{lstlisting}

The receiver's top-level application is dual and similar to the
sender. In the receiver's function `\ls$receive_file file c$', \ls$file$ starts off
as a potentially uninitialized buffer, but the postcondition
of \ls$receive_file$ guarantees that on a successful run,
the \ls$file$ is partially filled with messages from a file that was
previously sent on the same connection, i.e., file transfer is
authentic.
\begin{lstlisting}
let received (#n:nat) (file:array byte n) (c:receiver) (h:heap) = file `initialized_in` h /\ sent (as_seq file h) c
val receive_file (#n:nat) (file:array byte n) (c:receiver{disjoint c file}) : MST (option nat)
   (ensures (fun h$_0$ ropt h$_1$ -> modifies {file, c} h$_0$ h$_1$ /\ (match ropt with None -> True
                                                  | Some r -> r <= n /\ received (prefix file r) c h$_1$)))
\end{lstlisting}

\paragraph*{Confidentiality.} To prove that the file transfer is confidential,
we relate the ciphers and tags sent on the network by two runs of the
sender application and prove that for arbitrary files that contain the
same number of messages, the network traffic is (probabilistically)
indistinguishable. As such, our file-transfer application is only
partially length hiding---the adversary learns lower and upper bounds
on the size of the file based on the number of messages it contains.
Our proof relies on a form of probabilistic
coupling~\cite{lindvall2002lectures}, relating the randomness used in
two runs of the sender by a bijection chosen to mask the differences
between the two files.
\iffull
\else
The proof technique used is independent of monotonicity and, as such,
is orthogonal to this paper.
\fi

\section{Ariadne: State continuity vs. hardware crashes}
\label{sec:case-studies.ariadne}

We present a second case study of monotonic state at work, this time
verifying a protocol whose very purpose is to ensure a form of
monotonicity.
We may improve the resilience of systems confronted by unexpected
failures (e.g., power loss) by having them persist their state
periodically, and by restoring their state upon recovery.
In this context, \emph{state continuity} ensures that the state is
restored from a recent backup, consistent with the observable effects
of the system, rather than from a stale or fake backup that an
attacker could attempt to replay or to forge.

For concreteness, we consider state continuity for Intel SGX
enclaves. Such enclaves rely on a special CPU mode to protect some
well-identified piece of code from the rest of the platform, notably
from its host operating system. As long as the CPU is powered, the
hardware automatically encrypts and authenticates every memory access,
which ensures state confidentiality and integrity for the protected
computation. For secure databases and many other applications, these
guarantees must extend across platform crashes so as to ensure, for instance, 
that any commit that has been reported as complete (\EG by witnessing
its result) remains committed once the system resumes.

Several recent papers address this problem~\cite{ParnoLDMM11,matetic17rote,StrackxP16}.
In this section, we present a first mechanized proof of
\citeauthor{StrackxP16}'s \emph{Ariadne} protocol using \fstar,
relying on our libraries for monotonic state---the proof essentially
consists of supplementing a hardware counter with a ghost state
machine to track the recovery process (naturally expressed as a
preorder) and then typechecking the recovery code.
This example also
illustrates the fairly natural combination of
monotonic state with other effects like exceptions.

\paragraph{The Ariadne Protocol.}
The protocol relies on a single, hardware-based, non-volatile
monotonic counter. Incrementing the counter is a privileged but
unreliable operation. (It is also costly, hence the need to minimize
increments.)  If the operation returns, then the counter is incremented; if it
crashes, the counter may or may not have been incremented.
We model this behavior using \ls$MSTExn$, a combination of
the \ls$MST$ monad and exceptions.
The result of an \ls$MSTExn$ computation
is either a normal result \ls$V v$ or an exception \ls$E e$.
In practice, our
protocol code never throws or catches exceptions---these
operations are available to the context only to model malicious
hosts.

\begin{lstlisting}
val incr: c:counter -> MSTExn unit
 (ensures fun h$_0$ r h$_1$ -> if r=V() then ctr c h$_1$ == ctr c h$_0$ + 1 else (ctr c h$_1$ == ctr c h$_0$ \/ ctr c h$_1$ == ctr c h$_0$ + 1))
\end{lstlisting}

The protocol also relies on a persistent but untrusted backup,
modelled as a host function \ls$save$.
%
Before saving the backup, the current state is encrypted and authenticated
together with a sequence number corresponding to the anticipated next
value of the counter.
The implementation of this authenticated encryption construction
(using the \ls$auth_encrypt$ function below) is similar to what is
presented in the \S\ref{sec:protocol_impl}; it makes use of
a hardware-based key available only to the protected code.
Conversely, the host presents an encrypted backup for recovery.
Recovery from the last backup always succeeds, although it may be
delayed by further crashes. Recovery from any other backup may fail,
but must not break state-continuity; we focus on verifying this safety 
property of Ariadne.

Below we give the code for creating an Ariadne-protected reference
cell and the two main operations for using it: \ls`store` and \ls`recover`.
Their pre- and postconditions are given later, but at a high level, \ls`store` requires
a good state and \ls`recover` does not.
(There is no separate \ls$load$ from a good state, since
the enclave can keep the state in volatile memory once it completes recovery.)
For simplicity, we model hardware protection using a private datatype
constructor: `\ls$Protect c k$' has type \ls$protected$ (defined
shortly) and packages the enclave capabilities to use the monotonic
counter~\ls$c$ and backup key~\ls$k$.
The \ls$recover$ function also takes as argument a (purportedly)
\ls$last_saved$ backup, which is first authenticated and decrypted: if
the decryption yields \ls $(m:nat,w:state)$ and \ls$m$ matches the
current counter \ls$!c$, then recovery continues with
state \ls$w$; otherwise it returns \ls$None$.

 \cf{We could cut \ls$create$ and compress the code, at some cost in
   readability.}\ch{Saved a line, should still be readable. Let's keep create.}
\da{The \fstar code has a call to witness in create, and no use of the 
assumed save function in this text.}
\begin{lstlisting}
let create (w:state) = let c = ref 0 in let k = keygen c in save (auth_encrypt k 0 w); Protect c k
let store (Protect c k) (w:state) = save (auth_encrypt k (!c+1) w); (*1*) incr c
let recover (Protect c k) last_saved =
  match auth_decrypt k !c last_saved with | None -> None | Some (w:state) ->
  save (auth_encrypt k (!c+1) w); (*2*) incr c; (*3*) save (auth_encrypt k (!c+1) w); (*4*) incr c; Some w
\end{lstlisting}

In this code, any particular call to \ls$save$ or \ls$incr$ may fail.
To update the state, `\ls$store$' first encrypts and saves a backup of
its new state (associated with the anticipated next value of the counter)
and then increments the counter.
This ordering of \ls$save$ and \ls$incr$ is necessary to enable recovery if a crash occurs at
point \ls$(*1*)$ between these two operations, but may provide the host
with several valid backups for recovery.
For instance, a malicious host may obtain encryptions of both \ls$v$ and \ls$w$ at
both \ls$!c$ and \ls$!c+1$ by causing failure at \ls$(*1*)$; recovering from
the older backup; causing failure at \ls$(*4*)$; then recovering and
causing failure at \ls$(*2*)$.
This is fine as long as the recovery process eventually commits to
either \ls$v$ or \ls$w$ before returning it.
To this end, recovery actually performs \emph{two} successive counter
increments to clear up any ambiguity, and to ensure there is a unique
backup at the current counter and no backup at the next counter.
(On the other hand, completing recovery after its first increment
would break continuity: with backups \emph{at the current counter} for
both \ls$v$ and \ls$w$, the host has ``forked'' the enclave and can
indefinitely get updated backups for both computations.)
To capture this argument on the intermediate steps of the
protocol,
we supplement the real state of the counter (modeled here as an integer) with one
out of 4 ghost cases, listed below.
\begin{lstlisting}
type case =
  | Ok:        saved:state -> case (* clean state at the end of store or recover *)
  | Recover:    read:state -> other:state -> case (* at step (3) above *)
  | Writing: written:state -> old:state -> case (* at steps (1) and (4) above *)
  | Crash:      read:state -> other:state -> case (* worst case at step (2) of recovery, outlined above *)
\end{lstlisting}


\ch{Would a better name for \ls$in_between$ be \ls$one_of$?}\cf{I slightly prefer in between}

Next, we give the specification of the protocol code, with
`\ls$ghost$' selecting the current ghost \ls$case$ of a counter, and
\ls@g `in_between`$ $ (v,w)@ stating that the ghost case \ls$g$ holds at
most \ls$v$ and \ls$w$.
\begin{lstlisting}
val create: w:state -> MSTExn protected (ensures fun h$_0$ r h$_1$ -> V? r ==> ghost h$_1$ (V?.v r) == Ok v)

val store: p:protected -> w:state -> MSTExn unit (requires fun h$_0$ -> Ok? (ghost h$_0$ p))
  (ensures fun h$_0$ r h$_1$ -> match r with
                     | V () -> ghost h$_1$ p == Ok w
                     | E _ -> ghost h$_1$ p `in_between` (Ok?.saved (ghost h$_0$ p),w))
                     
val recover: p:protected -> backup (Protect?.c p) -> MSTExn (option state)
 (ensures fun h$_0$ r h$_1$ -> exists v w. ghost h$_0$ p `in_between` (v,w) /\ match r with
                                                     $\hspace{1.8ex}$  | V (Some u) -> ghost h$_1$ p == Ok w /\ u == w
                                                     $\hspace{1.8ex}$  | _  $\hspace{8.9ex}$  -> ghost h$_1$ p `in_between` (v,w))
\end{lstlisting}

To verify this specification, we introduce a `\ls$saved$' predicate that
controls, in any given counter state (\ls$n,g$), which backups
(\ls$seqn,u$) may have been encrypted and saved, and thus may be
presented by the host for recovery.
We use this predicate both to define our preorder on counter states, and to refine
the type \ls$nat*state$ of general authenticated-encrypted values so as to define the type 
\ls$backup c$ of 
authenticated, encrypted, and saved backups associated with a given counter state \ls$c$. 
\begin{lstlisting}
let saved (n:nat,g:case) (seqn:nat,u:state) = (* overapproximating what may have been encrypted and saved *)
  (seqn < n) \/ (* an old state; authentication of `seqn` will fail, so nothing to say about it *)
  (seqn == n$\hspace{2ex}$   /\ (match g with | Ok v -> u == v
                              $\hspace{-0.06cm}$| Recover w v | Writing w v | Crash w v -> u == w \/ u == v)) \/
  (seqn == n+1 /\ (match g with | Ok _ | Recover _ _ -> False
                              $\hspace{-0.015cm}$| Writing v _ -> u == v
                              $\hspace{-0.015cm}$| Crash v w -> u == v \/ u == w))

let preorder (n$_0$,g$_0$) (n$_1$,g$_1$) = forall s. saved (n$_0$,g$_0$) s ==> saved (n$_1$,g$_1$) s
let saved_backup c s h = h `contains` c /\ saved (h.[c]) s
type backup c = s:(nat*state){witnessed (saved_backup c s)}

(* append$\text{-}$only ghost logs used to model the history of saving backups; we attach one to every backup key `k` *)
type log c = mref (list (backup c)) (fun l$_0$ l$_1$ -> l$_0$ `prefix_of` l$_1$) 
type protected = Protect: c:mref (nat*case) preorder -> k:key c -> protected
\end{lstlisting}

Our method for verifying the protocol is to use monotonicity to capture 
the history of saving backups, by augmenting 
the \ls$save$ function with \ls$witness (saved_backup c s)$
for the given authenticated-encrypted backup \ls$s$, and 
the \ls$recover$ function with \ls$recall (saved_backup c last_saved)$ just after 
authenticated decryption of the (purportedly) \ls$last_saved$ backup the host 
provides for recovery. This allows us to 
recover at which stage of the protocol \ls$last_saved$ might
have been created. We also instrument key points in
the code with erasable modeling functions that operate on the ghost counter  
state. Typechecking enforces that all updates respect the
preorder that ties the ghost state to the actual counter value.
For instance, \ls$incr$ has two ghost transitions:
\ls@(n, Crash w v) $\leadsto$ @\linebreak[2]\ls@(n+1, Recover w v)@
for typing call  \ls$(*2*)$ above,
and
\ls@(n, Writing w v)$\leadsto$ (n+1, Ok w)@
for typing calls \ls$(*1*)$ and \ls$(*4*)$.
The other transitions update just the ghost case of the counter state.
For instance, the only transition that introduces a new \ls$w:state$ is
\ls@(n,Ok v) $\leadsto$ (n,Writing w v)@,
used in the \ls$store$ function before the call to encrypt the new backup.
This yields a concise \fstar proof of the safety 
of Ariadne's state continuity, possibly simpler than the original paper
proof of \citeauthor{StrackxP16}.

\ch{Is this the thing we could draw? In general giving more intuition
  could help understanding this section.}

\cf{For simplicity, we do not consider the functional proof, showing
  that recovery from the last save either crashes or returns sone
  recovered state.}

\cf{Commented out discussion on how we may use the code above to
  conduct enclave-level proofs involving monotonicity.}






\section{Meta-theory of the abstract monotonic-state monad}
\label{sect:metatheory}

While the applications of monotonic state are many, diverse, and
sometimes quite complex, in the previous sections we showed how they can all
be reduced to our general monotonic-state monad from \autoref{sec:mst}
(\ls$MST$ in \autoref{fig:mst} on page \pageref{fig:mst}).
In this section and the next, we put reasoning with this
monotonic-state monad on solid foundations by presenting its
meta-theory in two stages.

We first present a calculus called \MSTcalculus to validate
the soundness of \ls$witness$ and \ls$recall$ in a setting in
which \ls$MST$ is treated abstractly.
%
%
In fact, \MSTcalculus is a fragment of another calculus,
\rMSTcalculus, which we study in the next section (\autoref{sect:reify}),
and that extends \MSTcalculus with support for safe, controlled monadic
reflection and reification for revealing the representation
of \ls$MST$.
Both calculi are designed to only capture the essence of reasoning
with the monotonic-state monad---thus they do not include
other advanced features of dependently typed languages, e.g., full
dependency, refinement types, inductives, and universes.
The complete definitions and proofs are available in the supplementary
materials at \url{https://fstar-lang.org/papers/monotonicity}




\ch{There is quite a bit of standard stuff coming up, so might want to
  explain upfront that there is something insightful coming up
  eventually to keep people motivated when reading through this}



\subsection{\MSTcalculus: A Dependently Typed Lambda Calculus with Support for \ls$witness$ and \ls$recall$}
\label{sec:syntaxofcalculus}

The syntax of \MSTcalculus is inspired by Levy's
fine-grain call-by-value language~\citep{LevyCBPV} in that it makes a
clear distinction between values and computations.
This set-up enables us to focus on stateful computations and
validating their correctness, in particular
for $\mathsf{witness}$ and $\mathsf{recall}$ actions.


First, \emph{value types} $t$ and \emph{computation types} $C$ are given by
the following grammar:
%
\[
\begin{array}{l c l}
t & ::= & \statet
~ \mid ~ \unit
~ \mid ~ t_1 \times t_2
~ \mid ~ t_1 +\, t_2
~ \mid ~ x{:}t -> C
\end{array}
\qquad\hspace{-0.15cm}
\begin{array}{l c l}
C & ::= & \mst t {s.~ \varphi_{\text{pre}}}
                 {s~x~s'.~\varphi_{\text{post}}}
\end{array}
\]
Here, $\statet$ is a fixed abstract type of states, and $x{:}t -> C$ is a value-dependent
function type, where the computation type $C$ depends on values of type $t$.
Similarly to Hoare Type Theory
and \fstar{} surface syntax, \MSTcalculus computation types
$\mst t {s.~ \varphi_{\text{pre}}} {s~x~s'.~\varphi_{\text{post}}}$
are indexed by pre- and postconditions:
%
%
the former are given over initial state
$s$, and the latter over
initial state $s$, return value $x$, and final state $s'$.
As a convention, we use variables $s$, $s'$, $s''$, \ldots{} to stand for states.



Next, \emph{value terms} $v$ and \emph{computation terms} $e$ are given by the following grammar:
\[
\begin{array}{l c l}
v,\sigma & ::= & x
 ~ \mid ~ c \in S
 ~ \mid ~ ()
 ~ \mid ~ (v_1,v_2)
 ~ \mid ~ \inl v
 ~ \mid ~ \inr v
 ~ \mid ~ \fun x t e
 \\[1mm]
 e & ::= & \return v
 ~ \mid ~ \bind{x}{e_1}{e_2}
 ~ \mid ~ v_1 ~ v_2
\\
& \mid & \pmatch v x {x_1} {x_2} e
 ~ \hspace{0.135cm} \mid ~
 \caseof v x {x_1} {e_1} {x_2} {e_2}
\\
& \mid & \stget
 ~ \mid ~ \stput \sigma
 ~ \mid ~ \witness {s . \varphi}
 ~ \mid ~ \recall {s . \varphi}
\end{array}
\]

Value terms are mostly standard, with $S$ denoting a fixed non-empty set of
$\statet$-valued constant symbols.%
{}
As a convention, we use $\sigma$, $\sigma'$, \ldots{} to
stand for value terms of type $\statet$.
%


Computation terms include returning values ($\return v$), sequential composition (${\bind{x\!}{\!e_1}{e_2}}$),
function applications ($v_1~v_2$), and
pattern matching for products ($\mathsf{pmatch}$) and sums
($\mathsf{case}$). In addition, computations include $\mathsf{get}$ and $\mathsf{put}$
actions for accessing the state and, finally,
$\mathsf{witness}$ and $\mathsf{recall}$ actions
parameterized by a first-order logic predicate ${s. \varphi}$
(where the variable $s$ is bound in $\varphi$).


Formulas ($\varphi$) used for the pre- and postconditions,
as well as in $\mathsf{witness}$ and $\mathsf{recall}$,
are drawn from a classical first-order logic
extended with (1) a fixed preorder on states (\rel),
(2) typed equality on value terms ($==$),
and (3) logical capabilities stating
that a predicate was true in some prior program state and can be recalled
in any reachable state ($\witnessed  {s. \varphi}$). The notation $s. \varphi$ defines a \emph{predicate}
on states by binding the variable $s$ of type $\statet$ in the formula $\varphi$.
\[
\begin{array}{l c l}
p & ::=  & \rel ~ \mid ~  ==\\
\varphi & ::= & p ~ \vec{v}
          ~ \mid ~ \top
          ~ \mid ~ \varphi_1 \wedge \varphi_2
          ~ \mid ~ \bot
          ~ \mid ~ \varphi_1 \vee \varphi_2
          ~ \mid ~ \varphi_1 => \varphi_2
          ~ \mid ~ \forall x \!:\! t .~ \varphi
          ~ \mid ~ \exists x \!:\! t .~ \varphi
          ~ \mid ~ \witnessed  {s . \varphi}
\end{array}
\]

We take formulas to be in first-order logic (as opposed to
\ls$Type$-valued functions in~\fstar) in order to
keep the meta-theory simple and focused on stateful computations.
Furthermore, this approach enables us to easily establish proof-theoretic
properties of the logical \ls$witnessed$-capabilities via a
corresponding sequent calculus presentation (see
\autoref{sec:inversion}). We use a classical logic to be
faithful to \fstar{}'s SMT-logic of pre- and postconditions shown in
our examples.

\subsection{Static Semantics of \MSTcalculus}

\ch{If space is tight I would leave out all the explanations of super
  boring wf judgments. In fact initially I had dropped those altogether
  from all rules.}

We define the type system of \MSTcalculus using judgments of
well-formed contexts (${|- \Gamma ~ \wf}$), 
well-formed logical formulas (${\Gamma |- \varphi ~ \wf}$), 
well-formed value types (${\Gamma |- t ~ \wf}$),
well-typed value terms (${\Gamma |- v : t}$), well-formed computation types
(${\Gamma |- C ~ \wf}$), and well-typed computation terms (${\Gamma |- e : C}$).
The rules defining the first five judgments are completely standard---we thus omit most of them and only give the formation rule for computation types:
\[\setnamespace{0pt}\setpremisesend{5pt}
\begin{array}{c}

\inferrule*[left=(Ty-MST)]
{
  \Gamma |- t ~ \wf \quad
  \Gamma, s{:}\statet |- \pre ~ \wf \quad
  \Gamma, s{:}\statet, x{:}t, s'{:}\statet |- \post ~ \wf
}
{
  \Gamma |- \mst{t}{s.~\pre}{s~x~s'.~\post} ~ \wf
}
\end{array}
\]
%
The typing rules for computation terms are more interesting---we present
a selection of them in \autoref{fig:first-typing}.
The rules for the
$\mathsf{put}$, $\mathsf{witness}$, and $\mathsf{recall}$ actions  correspond
directly to the typed interface from \autoref{fig:mst} on page
\pageref{fig:mst}.
%
%
Similarly to \fstar{}, \MSTcalculus supports subtyping for value types
(${\Gamma |- t <: t'}$) and computation types (${\Gamma |- C <: C'}$).
%
Subtyping computation types (rule {\sc SubMST} from
\autoref{fig:first-typing}) is covariant in the postconditions, and
contravariant in the preconditions.
Moreover, logical entailment between the postconditions is proved assuming
that the preconditions hold, and that the initial and final states are
related by $\rel$ (recall that the state is only allowed to evolve according to $\rel$).

\begin{figure}

\[\setnamespace{0pt}\setpremisesend{5pt}
\hspace{-0.2cm}
\begin{array}{c}

\inferrule*[Left=\tightlabel{T-Return}]
{
  \Gamma |- v : t
}
{
  \Gamma |- \return v :
    \mst {t} {s . \top}
             {s~x~s'.~ s == s' \wedge x == v}
}

\hspace{0.5cm}

\inferrule*[left=(T-App)]
{
  \Gamma |- v_1 : x{:}t -> C \qquad
  \Gamma |- v_2 : t
}
{
  \Gamma |- v_1~v_2 : C[v_2/x]
}

\\[0.6em]

\inferrule*[lab=(T-Bind)]
{
  \Gamma |- e_1 :
    \mst {t_1} {s . \pre}
              {s~x_1~s'.~ \post}
  \quad
  \Gamma, x_1{:}t_1 |- e_2 :
    \mst {t_2} {s'.~ \exists s{:}\statet.~ \post}
               {s'~x_2~s''.~ \post'}
}
{
  \Gamma |- \bind{x_1}{e_1}{e_2} :
    \mst {t_2} {s. \pre}
               {s~x_2~s''.~ \exists x{:}t_1.~\exists s'{:}\statet.~
                              \post \wedge \post'}
}

\\[0.6em]

\inferrule*[Left=\tightlabel{T-Get}]
{
  |- \Gamma~\wf
}
{
  \Gamma |- \stget :
    \mst {\statet} {s . \top}
                   {s~x~s'.~ s == x == s'}
}

\hspace{0.2cm}

\inferrule*[Left=\tightlabel{T-Put}]
{
  \Gamma |- \sigma : \statet
}
{
  \Gamma |- \stput \sigma :
    \mst {\unit} {s.~ \rel~s~\sigma}
                   {s~x~s'.~ s' == \sigma}
}

\\[0.7em]

\inferrule*[Left=\tightlabel{T-Witness}]
{
  \hspace{2cm}
  \Gamma,s{:}\statet |- \varphi ~ \wf
  \qquad
  \stable{s.\varphi} =_\text{def} \forall s'.~\forall s''.~  \rel~s'~s'' \wedge \varphi[s'/s] ==> \varphi[s''/s]
}
{
  \Gamma |- \witness{s.\varphi} :
    \mst {\unit} {s'.~ \stable{s.\varphi} \wedge \varphi[s'/s]}
                 {s'~x~s''.~ s' == s'' \wedge  \witnessed{s.\varphi}}
}

\\[0.7em]

\inferrule*[Left=\tightlabel{T-Recall}]
{
  \Gamma,s{:}\statet |- \varphi ~ \wf
}
{
  \Gamma |- \recall{s.\varphi} :
    \mst {\unit} {s'.~ \witnessed{s.\varphi}}
                 {s'~x~s''.~ s' == s'' \wedge \varphi[s''/s]}
}

\\[0.7em]

\inferrule*[lab=(T-SubComp)]
{
  \Gamma |- e : C \quad
  \Gamma |- C <: C'
}
{
  \Gamma |- e : C'
}

\hspace{0.5cm}

\inferrule*[Left=\tightlabelup{Sub-MST}]
{
  \Gamma |- t <: t' \qquad
  \Gamma, s{:}\statet ~|~ \pre' |- \pre\\\\
  \Gamma, s{:}\statet, x{:}t, s'{:}\statet ~|~ \pre', \rel~s~s', \post |- \post'
}
{
  \Gamma |- \mst{t}{s.~\pre}{s~x~s'.~\post} <: \mst{t'}{s.~\pre'}{s~x~s'.~\post'}
}

\end{array}
\]
\caption{Selected typing and subtyping rules for \MSTcalculus}
\label{fig:first-typing}
\end{figure}

%
%
%
%
%

We define logical entailment for formulas $\varphi$ using a natural deduction system
($\Gamma ~|~ \Phi |- \varphi$, where $\Phi$ is a finite set of
formulas). Most rules are standard for classical first-order logic, so
\autoref{fig:logic} only lists the rules that are specific to our setting.
These include weakening for $\mathsf{witnessed}$, reflexivity and transport for
equality, reflexivity and transitivity for $\rel$, and rules about the equality of pair and
sum values. In more expressive logics, the latter rules are commonly derivable.

\begin{figure}

\[\setnamespace{0pt}\setpremisesend{5pt}
\hspace{-0.1cm}
\begin{array}{c}

\inferrule*[lab=(Witnessed-Weaken)]
{
  \Gamma, s{:}\statet ~|~ \Phi |- \varphi ==> \varphi'
  \quad
  s \not\in FV(\Phi)
}
{
  \Gamma ~|~ \Phi |- (\witnessed{s.\varphi}) ==> (\witnessed{s.\varphi'})
}

\hspace{0.3cm}

\inferrule*[lab=(Eq-Refl)]
{
  \Gamma |- v : t
}
{
  \Gamma ~|~ \Phi |- v == v
}

\hspace{0.3cm}

\inferrule*[lab=(Eq-Transport)]
{
  \Gamma ~|~ \Phi |- v == v'
  \quad
  \Gamma ~|~ \Phi |- \varphi[v/x]
}
{
  \Gamma ~|~ \Phi |- \varphi[v'/x]
}

\\[0.6em]

\inferrule*[lab=(Rel-Refl)]
{
  \Gamma |- \sigma : \statet
}
{
  \Gamma ~|~ \Phi |- \rel~\sigma~\sigma
}

\hspace{0.5cm}

\inferrule*[lab=(Rel-Trans)]
{
  \Gamma ~|~ \Phi |- \rel~\sigma_1~\sigma_2
  \\\\
  \Gamma ~|~ \Phi |- \rel~\sigma_2~\sigma_3
}
{
  \Gamma ~|~ \Phi |- \rel~\sigma_1~\sigma_3
}

\hspace{0.5cm}

\inferrule*[lab=(Sum-Disjoint)]
{
  \Gamma ~|~ \Phi |- \inl{v_1} == \inr{v_2}
}
{
  \Gamma ~|~ \Phi |- \bot
}

\hspace{0.5cm}

\inferrule*[lab=(Pair-Injective)]
{
  \Gamma ~|~ \Phi |- (v_1,v_2) == (v'_1,v'_2)
}
{
  \Gamma ~|~ \Phi |- v_i == v'_i
}

%

\end{array}
\]

\caption{Natural deduction (selected rules)}
\label{fig:logic}
\end{figure}

\subsection{Instrumented Operational Semantics of \MSTcalculus}
\label{sec:first-opsem}

We equip \MSTcalculus with a small-step operational semantics that we
instrument with a log of witnessed stable properties, which we use as
additional logical assumptions in the correctness theorem we prove in
\autoref{sec:correctness}, so as to accommodate the $\mathsf{witness}$
and $\mathsf{recall}$ actions, and their typing.
%
Formally, we define the reduction relation $\leadsto$
on \emph{configurations} $(e,\sigma,W)$,
where $e$ is the computation term being reduced,
$\sigma$ is the current state value,
and $W$ is a finite set logging the stable predicates $s.\varphi$
witnessed so far.
Most reduction rules are standard, so we only list the rules for
actions below:
\[
\begin{array}{rrcl}
\textsc{\small(Get)} & (\stget,\sigma,W) & \leadsto & (\return{\sigma},\sigma,W) \\
\textsc{\small(Put)} & (\stput{\sigma'},\sigma,W) & \leadsto & (\return{()},\sigma',W) \\
\textsc{\small(Witness)} &
(\witness{s.\varphi},\sigma,W) & \leadsto & (\return{()},\sigma,W \cup s.\varphi)\\
\textsc{\small(Recall)} & (\recall{s.\varphi},\sigma,W) & \leadsto & (\return{()},\sigma,W)
\end{array}
\]
The $\mathsf{witness}$ action adds the witnessed
stable predicate to the log $W$, while $\mathsf{get}$, $\mathsf{put}$, and
$\mathsf{recall}$ do not use the log at all. The $\mathsf{get}$ action
returns the current state; the $\mathsf{put}$ action overwrites it.

\subsection{Correctness for \MSTcalculus}
\label{sec:correctness}

We now prove the correctness of our instrumented operational semantics
in a Hoare-style program logic sense. In particular,
we show that if ${~|- e : \mst{t}{s.~\pre}{s~x~s'.~\post}}$,  $e$ reduces
to $\return v$, and $\pre$ holds of the initial state, then $\post$
holds of the initial state, the returned value $v$, and the final state.
We also establish that the initial state is related
to the final state, and that the logs can ever only increase.
In order to better structure our proofs, we split them into
progress and preservation theorems, which when
combined give the above-mentioned result in \autoref{thm:partialcorrectness}.

%

A key ingredient of the following correctness results is the use of the instrumentation
(log $W$ of witnessed stable predicates) to provide additional logical assumptions
corresponding to the logical \ls$witnessed$ capabilities resulting from earlier \ls$witness$
actions. In detail, these additional logical assumptions take the form
$\witnessed W =_\text{def} \witnessed s'.(\bigwedge_{s.\varphi \in W} \varphi[s'/s])$.
In the results below, we also use well-formed state-log pairs
$\Gamma |- (\sigma,W) ~ \wf$, defined to hold iff ${\Gamma |- \sigma : \statet}$ and 
\[
  \Gamma ~|~ \witnessed W |- \varphi[\sigma/s] \qquad
  \Gamma ~|~ \witnessed W |- \stable{s.\varphi} \qquad
  (\text{for all } s.\varphi \in W)
\]
%
%

%



\begin{theorem}[Progress]
If $~|- e : \mst{t}{s.~\pre}{s~x~s'.~\post}$ then either $\exists v . ~ e = \return v$ or
$\forall \sigma ~ W . ~ \exists e' ~ \sigma' ~ W' . ~ (e,\sigma,W) \leadsto (e',\sigma',W')$.
\end{theorem}




%

Preservation crucially uses the well-formedness of $(\sigma,W)$
to record that all previously witnessed stable predicates are in fact true of
the current state, in combination with the $\witnessed W$ assumptions which ensure
that, once obtained, logical capabilities remain usable in the future.



\begin{theorem}[Preservation]
\label{thm:preservation}
If
${~|- e : \mst{t}{s.~\pre}{s~x~s'.~\post}}$ and
${(e,\sigma,W) \leadsto (e',\sigma',W')}$ such that
${|- (\sigma,W) ~ \wf}$ and
${\witnessed W |- \pre[\sigma/s]}$, then
\begin{enumerate}
\item ${W \subseteq W'}$ and ${\witnessed {W'} |- \rel ~ \sigma ~ \sigma'}$
  and ${|- (\sigma',W') ~ \wf}$  and
\item $\exists \pre' ~ t' ~ \post' .$
${|- e' : \mst{t'}{s.~\pre'}{s~x~s'.~\post'}}$ and
${|- t' <: t}$ and
${\witnessed {W'} |- \pre'[\sigma'/s]}$ and
${x{:}t' , s'{:}\statet ~|~ \witnessed {W'} , \rel ~ \sigma' ~ s' , \post'[\sigma'/s] |- \post[\sigma/s]}$.
\end{enumerate}
\end{theorem}

\begin{proof}
By induction on the sum of the height of the derivation of ${(e,\sigma,W) \leadsto (e',\sigma',W')}$
and the size of the computation term $e$, and by inverting the judgment
${|- e : \mst{t}{s.~\pre}{s~x~s'.~\post}}$ for
each concrete $e$. Below we comment briefly on the more interesting cases of this proof.

\vspace{0.15cm}

\noindent
{\sc{Put}}:
In this case, ${e = \stput \sigma''}$, ${\sigma' = \sigma''}$, and ${W' = W}$,
and we prove $\witnessed W' |- \rel ~ \sigma ~ \sigma'$ by combining the assumption
$\witnessed W |- \pre[\sigma/s]$ with $s{:}\statet ~|~ \pre |- \rel ~ s ~ \sigma''$
that we get by inverting the typing of $\stput \sigma''$.
We then derive $|- (\sigma',W') ~ \wf$ from $|- (\sigma,W) ~ \wf$ by using the stability of
the witnessed predicates, proving
$\witnessed W |- \varphi[\sigma/s] ==> \varphi[\sigma''/s]$ for all $s.\varphi \in W$.


\vspace{0.15cm}

\noindent
{\sc{Witness}}:
In this case, ${e = \witness{s.\varphi}}$, ${\sigma' = \sigma}$,
${W' = W \cup s.\varphi}$, and $\pre' = \witnessed s.\varphi$,
and we prove $\witnessed W' |- \pre'[\sigma'/s]$ by using
the {\sc{Witnessed-Weaken}} rule from \autoref{fig:logic}.

\vspace{0.15cm}

\noindent
{\sc{Recall}}:
In this case, ${e = \recall{s.\varphi}}$, ${\sigma' = \sigma}$, ${W' = W}$, and
${\pre' = \varphi}$, and we are required to \linebreak prove ${\witnessed W' |- \pre'[\sigma'/s]}$.
We do so by combining the assumption ${|- (\sigma,W) ~ \wf}$ with \linebreak the proof of
${s{:}\statet ~| \bigwedge_{s'.\varphi' \in W} \varphi'[s/s'] |- \varphi}$,
which follows from the assumed precondition \linebreak ${\witnessed W |- \witnessed s.\varphi}$;
this is a proof-theoretic property of the logic, see \autoref{sec:inversion} for details.
\end{proof}


Finally, we combine progress and preservation results to prove partial correctness for \MSTcalculus.




\begin{proposition}[Correctness of $\mathsf{return}$]
If ${|- \return v : \mst{t}{s.~\pre}{s~x~s'.~\post}}$ and \linebreak ${|- \sigma : \statet}$, then
$|- v : t$ and $|- \pre[\sigma/s] ==> \post[\sigma/s,v/x,\sigma/s']$.
\end{proposition}



\begin{theorem}[Partial correctness]
\label{thm:partialcorrectness}
If ${~|- e : \mst{t}{s.~\pre}{s~x~s'.~\post}}$ and we have a reduction sequence
${(e,\sigma,W) \leadsto^{*} (\return v,\sigma',W')}$ such that $|- (\sigma,W) ~ \wf$
and $\witnessed W |- \pre[\sigma/s]$, then $W \subseteq W'$ and
$\witnessed {W'} |- \rel ~ \sigma ~ \sigma'$, $|- v : t$ and
$\witnessed W' |- \post[\sigma/s,v/x,\sigma'/s']$.
\end{theorem}

We strengthen this to total correctness by also showing that
\MSTcalculus is strongly normalizing.

\begin{theorem}
If $\Gamma |- e : C$ and $|- (\sigma,W) ~ \wf$, then $(e,\sigma,W)$ is strongly
normalizing in \MSTcalculus.
\end{theorem}

\begin{proof}
The proof is based on defining
a typing- and reduction structure preserving translation~$| \text{-} |$
of \MSTcalculus types, terms, and configurations
to a corresponding strongly normalizing simply typed
calculus (by erasing type dependency,
logical formulas, and logs, e.g., $| x{:}t -> C| =_{\text{def}} |t| -> |C|$,
$| \witness{s.\varphi} | =_{\text{def}} \mathsf{witness}$\footnote{To
simplify the proof, the simply typed calculus
includes computationally irrelevant computation terms
$\mathsf{witness}$ and $\mathsf{recall}$.}, and $|(e,\sigma,W)| =_{\text{def}} (|e|,|\sigma|)$).
We establish the strong normalization of this simply typed calculus
using the standard $\top\top$-lifting approach \citep{Lindley2005}.
\end{proof}



\ch{In the worst case, would we be up to dropping or shrinking all proofs?}

\subsection{Proof-theoretic Properties of the Logic of Pre- and Postconditions}
\label{sec:inversion}

\nik{Why didn't we use the sequent calculus in the definition right away, instead of
having to switch from natural deduction here? Needs some more motivation}
\ch{If space is very tight could we consider dropping the natural
  deduction formulation?}

We conclude our investigation into the meta-theory of \MSTcalculus by recalling that
in the proof of \autoref{thm:preservation}, it was crucial (in the {\sc Recall} case) to
construct a derivation of $s{:}\statet ~|~ \varphi |- \varphi'$ from a derivation
of $\witnessed{s.\varphi} |- \witnessed{s.\varphi'}$. Intuitively, this proof-theoretic property of
the logic means that $\varphi'$ must be a logical consequence of $\varphi$ because
$\witnessed{s.\varphi} |- \witnessed{s.\varphi'}$ could only have been
proved using the natural deduction hypothesis and {\sc{Witnessed-Weaken}}
rules.

It is well known that establishing such properties directly in
natural deduction is difficult due to the introduction rule
for implication. We follow standard practice and turn to sequent
calculus, which we define using judgment $\Gamma ~|~ \Phi |- \Phi'$, where
$\Phi$ and $\Phi'$ are finite sets of formulas. Most rules are
standard for classical sequent calculus, so \autoref{fig:seqcalc} only
lists the more relevant ones. Importantly, the rules for $\Gamma ~|~ \Phi |- \Phi'$
do not include cut---we instead prove that it is admissible.

\begin{theorem}[Admissibility of cut]
The cut rule is admissible in this sequent calculus.
\end{theorem}

%
%

In order to accommodate the rules concerning $==$ and $\rel$ from \autoref{fig:logic},
while at the same time ensuring that cut remains admissible,
we follow \citet{Negri1998} in defining the corresponding rules in the sequent
calculus such that they only modify
the left-hand side of the entailment judgment---see the equality rules in
\autoref{fig:seqcalc}. Also following \citeauthor{Negri1998}, we restrict the
{\sc{Eq-Transport-SC}} rule to atomic predicates (the general rule
is admissible).

\begin{figure}

\[\setnamespace{0pt}\setpremisesend{5pt}
\begin{array}{c}

\inferrule*[lab=(Eq-Refl-SC)]
{
  \Gamma ~|~ \Phi , v == v |- \Phi'
  \quad
  \Gamma |- v : t
}
{
  \Gamma ~|~ \Phi |- \Phi'
}

\hspace{1cm}

\inferrule*[lab=(Eq-Transport-SC)]
{
  \Gamma ~|~ \Phi , (p ~ \vec{v})[v''/x] |- \Phi'
  \quad
  \Gamma |- v' : t
  \quad
  \Gamma |- v'' : t
}
{
  \Gamma ~|~ \Phi , v == v' , (p ~ \vec{v})[v'/x] |- \Phi'
}

\\[0.6em]

\inferrule*[left=(Witnessed-Weaken-SC)]
{
  \Gamma, s{:}\statet ~|~ \Phi , \varphi |- \varphi' , \Phi'
  \quad
  s \not\in FV(\Phi)
  \quad
  s \not\in FV(\Phi')
}
{
  \Gamma ~|~ \Phi , \witnessed{s.\varphi} |- \witnessed{s.\varphi'} , \Phi'
}

\end{array}
\]

\caption{Sequent calculus (selected rules)}
\label{fig:seqcalc}
\end{figure}

Next, we relate this sequent calculus to our natural deduction system in the standard
way, e.g., if $\Gamma ~|~ \Phi |- \Phi'$ in the sequent calculus, then
$\Gamma ~|~ \Phi |- \bigvee_{\varphi \in\, \Phi'} \varphi$ in natural deduction, and vice versa.
Based on the equivalence of the systems, and the syntax directed nature of the sequent
calculus rules, we finally prove the property of the logical $\mathsf{witnessed}$
capabilities we used in the proof of \autoref{thm:preservation}.



\begin{theorem}
If we have a derivation of $~\Gamma ~|~ \Phi , \witnessed{s.\varphi} |- \witnessed{s.\varphi'}$ in the sequent calculus
such that $s \not\in FV(\Phi)$ and $\Phi$ contains only formulas $p ~ \vec{v}$, then we have
$~\Gamma , s{:}\statet ~|~ \Phi , \varphi |- \varphi'$.
\end{theorem}

Another standard consequence of cut admissibility is logical \emph{consistency}:
the addition of the rules concerning \ls$witnessed$, $==$, and $\rel$ do
not make our natural deduction system inconsistent.



\section{Revealing the representation of the monotonic-state monad}
\label{sect:reify}

\ch{For full honesty, we would make it clear that part is more
  speculative than the rest of the paper. In particular, it wasn't yet
  implemented in \fstar{} and we don't really know whether it would
  work in practice. First experiments are not so encouraging.}

\ch{Should better explain the limitations of \autoref{sect:reify},
  in preparation for a proper solution to the problem.}
  
\da{I tried to add couple of sentences about this in the future 
  work section.}

So far, the key ingredient for ensuring that
the \ls$witness$ and \ls$recall$ actions are sound has been the abstract
treatment of \ls$MST$. In particular,
as illustrated in \autoref{sec:mst}, carelessly breaking the abstraction
and revealing \ls$MST$ computations as pure state-passing functions
quickly leads to unsoundness in the presence of \ls$witness$ and
\ls$recall$.
In this section, we present our full formal dependently typed lambda
calculus \rMSTcalculus (as
a delta relative to \MSTcalculus that we studied in \autoref{sect:metatheory}),
which also supports revealing the
representation of computations in a controlled and provably safe way.

\subsection{Reify and Reflect, by Example}
\label{sec:reifyreflectbyexample}

Our starting point is to introduce two coercions for soundly exposing
the representation of the monotonic-state monad:
\ls$reify$ for revealing the representation of an \ls$MST$ computation
as a state-passing function; and \ls$reflect$ that turns a
state-passing function into an \ls$MST$ computation \cite{Filinski94,
dm4free}.
Then, we carefully arrange the \rMSTcalculus calculus so that the
typing of computations that rely on the representation of \ls$MST$
effectively treat the \ls$witness$ and \ls$recall$ actions as no-ops.
By following this approach, we prevent the inconsistencies sketched in \S\ref{sec:mst}, yet
allow several interesting applications based on the pure state-passing
representation of \ls$MST$.



\paragraph{Reification is Useful.} One useful application of \ls$reify$ is
to prove relational properties of stateful programs \cite{relational}.
Consider a variant of the monotonic counter from \autoref{sec:mst}
that we can increment by an arbitrary positive amount.
We would like to be able to show that a simple computation branching
on a secret boolean \ls$h$ and then either incrementing the counter by
2 or incrementing it twice by 1 (see \ls$incr2$ below) does not leak
information about \ls$h$ into the public counter, a noninterference
property proven by the assertion on reified terms on the last line:
\begin{lstlisting}
let incr (n:nat) : MST unit = put (get() + n)
let incr2 (h:bool) : MST unit = if h then incr 2 else (incr 1; incr 1)
assert (forall h$_0$ h$_1$ c. reify (incr2 h$_0$) c = reify (incr2 h$_1$) c) (* noninterference *)
\end{lstlisting}
In particular, observe that by turning an abstract \ls$MST$ computation 
into a pure state-passing function, \ls$reify$ serves two purposes in this 
example: on the one hand, it enables one to apply the \ls$MST$ computation 
\ls$incr2$ to an initial state argument \ls$c$; on the other hand, it allows 
this \ls$MST$ computation to be used in the assertion where only pure 
(effect-free) programs are allowed in \fstar.

\paragraph{Reflection is Useful.} The dual of reification, monadic
reflection, enables users to extend effect interfaces with new actions
defined as pure state-passing functions.
For instance, suppose we want to add a new monadic action for squaring
the value of a monotonic counter.
The counter interface might not provide all the operations for
conveniently implementing this (\EG it might only provide \ls$incr$
and \ls$is_above n$), but if the interface exposes \ls$reflect$, then we can
easily implement the squaring action from scratch. e.g.,
\ls$let square () : MST unit = reflect (fun s -> ((), s * s))$.

\vspace{0.25cm}

We expect this kind of reasoning to be sound, certainly when, as in
the examples above, \ls$witness$ and \ls$recall$ are not used at all.
However, even slightly more complex examples implicitly rely
on \ls$witness/recall$, though this may not always be immediate,
e.g., programs using typed references implicitly
use \ls$recall$ at dereferencing.
Yet, allowing unrestricted combinations of
\ls$reify$, \ls$reflect$, \ls$witness$, and \ls$recall$
immediately enables counterexamples (\autoref{sec:mst}).
Our type system of \rMSTcalculus ensures that programs relying on
the representation of \ls$MST$ gain no benefit from
using \ls$witness$ and \ls$recall$.

\subsection{\rMSTcalculus: Syntax and Static Semantics}



The main idea of \rMSTcalculus is to simultaneously provide two
version of the \ls$MST$ computation type, one which is abstract, and
another whose representation is exposed.
As such, we propose an \emph{indexed computation type}, \ls@MST$_b$@, where
the boolean $b$ signifies whether or not the computation is
``representation aware'' ($b=\true$) or abstract ($b=\false$), i.e.,
the abstract \ls$MST$ computation type of~\autoref{sect:metatheory} is
really just \ls@MST$_\false$@.
Moreover, reusable \ls@MST$_b$@ computations can be parameterized over
an arbitrary boolean $b$, and so will be usable with both values of
$b$, similar to an intersection type.

The value types $t$ of \rMSTcalculus are identical to \MSTcalculus;
and we encode the type of booleans simply as $\boolt =_{\text{def}} \unit + \unit$.
As a convention, we use $b$ to range over value terms of type $\boolt$.
More interestingly, \rMSTcalculus's computation types $C$ are given
by the following extended grammar:
\[
\begin{array}{l c l}
C & ::= & \pure t {\pre} {x.~\post}
               ~\mid~
	       \bmst b t {s.~ \varphi_{\text{pre}}}
                            {s~x~s'.~\varphi_{\text{post}}}
\end{array}
\]
%
To give strong types to the state-passing functions involved
in \ls$reify$ and \ls$reflect$, we follow \fstar and introduce
\ls$Pure$, the type of pure, effect-free computations \cite{dm4free},
enabling us to reason about pure programs in terms of pre- and
postconditions using the same logic as in \autoref{sect:metatheory}.

The value and computation terms of \rMSTcalculus are also a minor
extension of \MSTcalculus's terms.  In particular, \rMSTcalculus also
includes a value term for monadic reification ($\reify e$), and a
computation term for monadic reflection ($\reflect v$). We also
provide a coercion from representation-aware to abstract computations
($\coerce e$, also a computation term), to be able to reuse a reflected
state-passing function in a context expecting an abstract computation.

As the monotonic-state monad is now indexed by a boolean $b$,
the \ls$get$, \ls$put$, \ls$witness$, and \ls$recall$ actions, and
the \ls$return$ of \ls@MST$_b$@, all take an additional boolean-valued
argument. Nevertheless, computation terms in \rMSTcalculus include
$\return v$ to return values in the \ls$Pure$ computation type.
\[
\begin{array}{l c l l c l}
v & ::= & \ldots
 ~ \mid ~ \reify e
 \\[1mm]
 e & ::= & \ldots
 ~ \mid ~ \breturn b ~ v
 ~ \mid ~ \bstget b
 ~ \mid ~ \bstput b ~ \sigma
 ~ \mid ~ \bwitness b ~ {s . \varphi}
 ~ \mid ~ \brecall b ~ {s . \varphi}
 ~ \mid ~ \reflect v
 ~ \mid ~ \coerce e
\end{array}
\]


\autoref{fig:second-typing} presents the main differences
to the typing rules of \rMSTcalculus relative to \MSTcalculus, where
parts typeset in gray apply only to
abstract \ls@MST$_\false$@ computations.
The \textsc{T-Witness} and \textsc{T-Recall} rules now
guard their pre- and postconditions with the boolean $b$.
As noted previously, for representation-aware computations ($b = \true$),
\ls$witness$ and \ls$recall$ are just no-ops, as seen by focusing only
on the non-gray parts of the rules.
Since representation-aware computations cannot make meaningful use of
monotonicity, the \textsc{T-Coerce} rule allows them to be
promoted to abstract \ls$MST$ computations---a coercion in the other
direction would immediately break soundness.
Finally, the postcondition of the state-passing function type
used in the \textsc{T-Reify} and \textsc{T-Reflect} rules
crucially captures the monotonic evolution of state.
For \ls$reify$, this is an important invariant that \ls$MST$
computations provide, which we can \emph{rely} on when their representation
is revealed. Dually, for \ls$reflect$, representation-aware code must
\emph{guarantee} to preserve an abstract computation's monotonic
view of the state.

\begin{figure}

\[\setnamespace{0pt}\setpremisesend{5pt}
\begin{array}{c}

\inferrule*[lab=(T-Witness)]
{
  \Gamma |- b : \boolt
  \qquad
  \Gamma,s{:}\statet |- \varphi ~ \wf
  \qquad
  {\color{gray}\stable{s.\varphi} =_\text{def} \forall s'.~\forall s''.~  \rel~s'~s'' \wedge \varphi[s'/s] ==> \varphi[s''/s]}
}
{
  \hspace{-1.5cm}
  \Gamma |- \bwitness b {s.\varphi} :
    {\color{dkblue}\mathsf{MST}}_{b}~\unit~
      (s'.~ {\color{gray} (b == \false) ==> (\stable{s.\varphi} \wedge \varphi[s'/s])})\\\\
      \hspace{4.275cm}
      (s'~x~s''.~ s' == s'' \wedge {\color{gray}((b == \false) ==>  \witnessed{s.\varphi})})
}

\\[0.6em]

\inferrule*[left=(T-Recall)]
{
  \Gamma |- b : \boolt
  \qquad
  \Gamma,s{:}\statet |- \varphi ~ \wf
}
{
  \hspace{-1.65cm}
  \Gamma |- \brecall b {s.\varphi} :
      {\color{dkblue}\mathsf{MST}}_{b}~\unit~
      (s'.~ {\color{gray}(b == \false) ==> \witnessed{s.\varphi}})\\
      \hspace{3.95cm}
      (s'~x~s''.~ s' == s'' \wedge {\color{gray}((b == \false) ==> \varphi[s''/s])})
}

\\[0.8em]

\inferrule*[left=(T-Coerce)]
{
  \Gamma |- e : \bmst {\true} t {s.~\pre} {s~x~s'.~\post}
}
{
  \Gamma |- \coerce e : \bmst {\false} t {s.~\pre} {s~x~s'.~\post}
}

\\[0.6em]

\inferrule*[left=(T-Reify)]
{
  \Gamma |- e :
     \bmst {\true} {t} {s.~\pre} {s~x~s'.~\post}
}
{
  \Gamma |- \reify e :
    s{:}\statet -> \pure{(t \times \statet)}{\pre}
                            {y.~\exists x .~ \exists s' .~ y == (x,s') \wedge \rel~s~s' \wedge \post}
}

\\[0.8em]

\inferrule*[left=(T-Reflect)]
{
  \Gamma |- v : s{:}\statet -> \pure{(t \times \statet)}{\pre}
                            {y.~\exists x .~ \exists s' .~ y == (x,s') \wedge \rel~s~s' \wedge \post}
}
{
  \Gamma |- \reflect v : \bmst {\true} t {s.~\pre} {s~x~s'.~\post}
}

\end{array}
\]
\caption{Selected typing rules for \rMSTcalculus}
\label{fig:second-typing}
\end{figure}


\subsection{Instrumented Operational Semantics of \rMSTcalculus}

We equip \rMSTcalculus with an instrumented operational semantics that
is a minor extension of  \MSTcalculus's semantics from \autoref{sec:first-opsem}.
The
reduction relation $\leadsto$ again works on \emph{configurations} $(e,\sigma,W)$.
We give the new reduction rules in \autoref{fig:second-reduction}.
Importantly, \rMSTcalculus has two rules for
the \ls$witness$ action:
\textsc{Witness-False} is just the \textsc{Witness}
reduction rule of \autoref{sec:first-opsem},
whereas \textsc{Witness-True} states that \ls$witness$ is a no-op in
representation-aware code.
The \textsc{Reify-Return} and \textsc{Reify-Context} rules make it
precise that $\reify e$ is a pure state-passing function that behaves
exactly like the given stateful computation~$e$, when applied to a
state $\sigma$. In particular, the \textsc{Reify-Return} rule
says that the reification of $\breturn {\true} {v}$ is essentially the same as
the pure state-passing function `$\lambda s{:}\statet .~ \return (v,s)$'; and
the \textsc{Reify-Context} rule takes care of \ls$get$, \ls$put$, etc., by
evaluating them as stateful computations.

The \textsc{Reflect-Return} and \textsc{Reflect-Context} rules
describe that $\reflect v$ behaves exactly like the
given pure state-passing function $v$.
In particular, in \textsc{Reflect-Context}
we evaluate the given function $\lambda s{:}\statet.~e$ by first substituting
the initial state $\sigma$ into its body, and by then reducing the resulting
pure computation term.
The computation term $e'$ (the reduct) is no longer dependent on the
lambda-bound variable, since the initial substitution of $\sigma$ into
$e$ removes all occurrences of the original bound variable.
\da{this explanation needs to be improved}
Finally, when the state-passing function finishes reducing
under \ls$reflect$, the \textsc{Reflect-Return} rule replaces the
current program state with the state computed by the pure
state-passing function. Together, \textsc{Reflect-Context}
and \textsc{Reflect-Return} ensure that a reflected computation's
effect on the state $\sigma$ occur atomically, i.e., the intermediate
states of pure, state-passing function are not observable and need not
respect the preorder, although, by \textsc{T-Reflect}, taken as a
single state transition, a reflected function is preorder compliant.

\da{comparison with EMF* semantics to be added when Kenji finishes the proof}

\begin{figure}
\[\setnamespace{0pt}\setpremisesend{5pt}
\hspace{-0.3cm}
\begin{array}{c}
\begin{array}{rrcl}
{\color{gray}\textsc{\small(Witness-False)}} &
{\color{gray}  (\bwitness \false s.\varphi ,\sigma,W)} & {\color{gray}\leadsto} & {\color{gray}(\breturn \false (),\sigma,W \cup s.\varphi)}
\\
\textsc{\small(Witness-True)} &
  (\bwitness \true s.\varphi ,\sigma,W) & \leadsto & (\breturn \true (),\sigma,W)
\\
\textsc{\small(Reify-Return)} &
  ((\reify (\breturn \true v)) ~ \sigma,\sigma',W) & \leadsto & (\return (v,\sigma),\sigma',W)
\\
\textsc{\small(Reflect-Reify)} &
  (\reflect (\reify e),\sigma,W) & \leadsto & (e,\sigma,W)
\\
\textsc{\small(Reflect-Return)} &
  (\reflect (\lambda s{:}\statet .~ \return (v,\sigma')),\sigma,W)
  & \leadsto &
  (\breturn \true v[\sigma/s],\sigma'[\sigma/s],W)
\\
\textsc{\small(Coerce-Return)} &
  (\coerce (\breturn \true v),\sigma,W)
  & \leadsto &
  (\breturn \false v,\sigma,W)
\end{array}
\\\\[-0.7em]
\begin{array}{c}
\inferrule*[left=(Reify-Context)]
{
  (e,\sigma,W) ~ \leadsto ~ (e',\sigma',W')
}
{
  ((\reify e) ~ \sigma,\sigma'',W) ~ \leadsto ~ ((\reify e') ~ \sigma',\sigma'',W')
}
\\[0.7em]
\inferrule*[left=(Reflect-Context)]
{
  (e[\sigma/s],\sigma,W) ~ \leadsto ~ (e',\sigma',W')
}
{
  (\reflect (\lambda s{:}\statet .~ e),\sigma,W)
  ~ \leadsto ~
  (\reflect (\lambda \_{:}\statet .~ e'),\sigma',W')
}
\end{array}
\end{array}
\]

\caption{Selected reduction rules for \rMSTcalculus}
\label{fig:second-reduction}
\end{figure}

\subsection{Correctness for \rMSTcalculus}
\label{sec:correctnessrMST}

We establish a partial correctness result for
\rMSTcalculus, similar to the one proved for \MSTcalculus in
\autoref{sec:correctness}.
We also structure our correctness proof using progress and
preservation theorems; as they are mostly analogous to
\autoref{sec:correctness}, we omit them here and only note the main differences:
(1) in order to use the well-typedness of the substitution $e[\sigma/s]$
for the \textsc{Reflect-Context} rule,
we consider only well-typed initial states in the progress theorem;
(2) \ls$Pure$ reduction steps do not modify the
state and the log of witnessed stable predicates; and
(3) \ls$MST$$_\true$ reduction steps to not change the log.

\begin{theorem}[Partial correctness, for Pure]
If ${~|- e : \pure{t}{\pre}{x.~\post}}$ and we have a reduction sequence
${(e,\sigma,W) \leadsto^{*} (\return v,\sigma',W')}$ such that $|- (\sigma,W) ~ \wf$
and $\witnessed W |- \pre$, then $W = W'$ and
$\sigma = \sigma'$ and $|- v : t$ and
$\witnessed W' |- \post[\sigma/s,v/x,\sigma'/s']$.
\end{theorem}

\begin{theorem}[Partial correctness, for MST$_b$]
If ${~|- e : \bmst{b}{t}{s.~\pre}{s~x~s'.~\post}}$ and
${(e,\sigma,W) \leadsto^{*} (\return v,\sigma',W')}$ such that $|- (\sigma,W) ~ \wf$
and $\witnessed W |- \pre[\sigma/s]$, then $W \subseteq W'$ (${W \!=\! W'}$ if ${~b\!=\!\true}$) and
$\witnessed {W'} |- \rel ~ \sigma ~ \sigma'$ and $|- v : t$ and
$\witnessed W' |- \post[\sigma/s,v/x,\sigma'/s']$.
\end{theorem}

\subsection{Discussion}
\label{sec:discussion}

We conclude by observing that, while \rMSTcalculus 
enables us to mix monotonic-state computations and 
reification-reflection in a sound way, the 
situation is not entirely satisfactory. 
Any 
\ls@MST$_b$@ computation that is parametric in 
$b$ essentially has to be equipped with 
two specifications, one making use of on monotonicity 
and employing the state-independent $\witnessed{s.\varphi}$ propositions, 
and the other using the corresponding predicates $s.\varphi$ as stateful 
invariants, so as to compensate for \ls$witness$ and \ls$recall$ being 
treated as no-ops when $b == $~\ls$true$. Instead, 
we expect it may be possible to avoid such duplication,
with abstract monotonic-state computations given a single
specification using state-independent $\witnessed{s.\varphi}$ 
propositions, as in \autoref{sect:metatheory}, from which the 
corresponding specification of the underlying representation (using
predicates $s.\varphi$ as stateful invariants) would then be 
derived automatically. We leave such a solution for future work
and further discuss it in relation to hybrid
modal logic in \autoref{sec:conclusion}.

\section{Related work}
\label{sec:relatedwork}

We have covered several strands of related work throughout the paper.
As discussed earlier, the closest related work to ours
is~\citepos{fstar-pldi13}, who propose and use in \fstar (but do not
formalize) a simple form of the witness/recall style that we
generalize, formalize, implement, and illustrate in this paper. However, reasoning about
monotonic state has a longer tradition---we discuss some closely
related, recent efforts below.

\vspace{-0.025cm}

\paragraph{The Essence of Monotonic State.}
\citet{PilkiewiczP11} devise a type system based on
linear and non-linear capabilities to model monotonic state.
Their \emph{fates} are linear capabilities constrained to evolve by a
preorder that are associated with, and describe invariants of,
fragments of the state. Due to their linearity, programs can gain
ownership of a fate, update the associated state, restoring
monotonicity without allowing intermediate, potentially non-monotonic
state transitions to be observable. Analogous to our stable
predicates, the non-linear capabilities in their system
(called \emph{predictions}) are properties that are stable with
respect to the evolution of a fate.
In contrast, our \ls$MST$ monad is based on a simpler, more abstract,
global state, deriving notions like typed references
which \citeauthor{PilkiewiczP11} take as primitive.
On the other hand,
our technique lacks any specific support for linearity or ownership;
it would be interesting to investigate this in the future.

\vspace{-0.025cm}

\paragraph{Rely-Guarantee References.} \citet{GordonEG13} devise a
monadic DSL in Coq with an \emph{axiomatic} heap and typed references,
where, like our monotonic references, a reference is associated with a
\emph{guarantee} relation constraining how its values may evolve. Whereas
in our `\ls$mref a rel$' the relation is a preorder
on \ls$a$, \citeauthor{GordonEG13} allow their guarantee relation to also
mention the entire heap.
This appears to be a middle ground between global preorders on the entire state
and the per-\ls$mref$
preorders.
%
\citeauthor{GordonEG13} also decorate references with a
\emph{rely} relation, analogous to our stable predicates, that are stable
with respect to guarantee.
They also include aliasing control whereby multiple
references to the same data can provide different, but compatible,
rely-guarantee relations.

\vspace{-0.025cm}

\paragraph{PCM-Based Separation Logics.}
Starting with the Fictional Separation Logic of \citet{Jensen12}, many
modern (concurrent) separation logics (e.g., \citepos{Krebbers0BJDB17}
Iris) use partial commutative monoids (PCMs) to allow users to define
their own problem-specific abstractions of the physical state.  These
PCM-based separation logics are very general, and can be used to
encode many styles of spatial and temporal reasoning about programs,
including reasoning about monotonic state. For example, \citeauthor{Jensen12} 
show how to use their logic to reason about a monotonic counter \ls$c$ using a 
(freely duplicable) predicate \ls$MC(c,i)$ from which one can conclude 
that the value of the counter \ls$c$ is at least \ls$i$, analogously to 
how we use the logical capability \ls$witnessed (fun c -> i <= c)$ 
to reason about monotonic counters in \autoref{sec:mst}.

This PCM-based model of monotonic counters is however somewhat 
imprecise, in that it only supports reasoning about monotonic counters 
using stable predicates, rather than freely mixing stable predicates with 
precise reasoning about the concrete value of the counter. In particular, 
after incrementing the counter \ls$c$, one can prove \ls$MC(c,i + 1)$ 
but not that the new value of the counter is exactly \ls$1$ 
greater than its previous value. In order to reason using both precise 
values and stable predicates, following discussions with an author of 
Iris, it seems necessary to move to a richer PCM that encodes a finer 
model of permissions, supporting both duplicable, stable observations 
as well as exclusively owned, precise assertions about the same resource, 
such as a monotonic counter. As such, even simple uses of monotonic 
reasoning in PCM-based logics seem to require some sophisticated encodings, 
even if staying in their sequential fragment. In contrast, we value the 
simplicity of our classical Hoare logic based verification model, 
recognizing that while PCMs are more general, our model is easy to 
use, benefits from SMT automation, and scales well.

\vspace{-0.025cm}

\paragraph{State-Machine Based Separation Logics.}
Some concurrent separation logics, such as CAP~\citep{Dinsdale-Young:2010} 
and its successors, use state machines to control 
how the program state is allowed to evolve.
For instance, \citet{SergeyWT18} recently devised a variant
of Hoare Type Theory for distributed protocols, that includes
primitive support for distributed ghost state governed by state
machines.
This is analogous to our use of preorders in the monotonic-state
monad: in \autoref{sec:protocol_iface} we show how distributed state
and state machines can be modeled using \ls$MST$.
However, all separation
logics that employ state machines of which we are aware have explicit
side-conditions on their reasoning rules requiring that all employed 
predicates are stable with respect to taking transitions in the state 
machine, ruling out mixing monotonic reasoning about stable predicates
with precise reasoning about non-stable ones when concurrent state
updates are not a concern.



\vspace{-0.025cm}

\paragraph{Reasoning About Object-Oriented Programs.}
Monotonicity is also of importance in the 
verification of object-oriented programs, where reflexive-transitive  
two-state invariants \citep{LeinoS07,CohenMST10} and update guards 
\citep{BarnettN04,PolikarpovaTFM14} are used to constrain how 
one is allowed to update an object, enabling one to 
establish that any consistent object depending on a consistent 
object \ls$O$ remains consistent after updates to \ls$O$. 
Also, similarly to our use of \ls$witness$ and \ls$recall$ to 
simplify the verification of stateful sequential programs, these works 
simplify the verification of object-oriented programs further by allowing
an object \ls$O$ to be explicitly marked as \emph{packed}, asserting that 
\ls$O$ is consistent and indicating that it remains so irrespective of updates 
to unrelated objects; and as \emph{unpacked}, indicating that \ls$O$ is 
susceptible to consistency-breaking updates. As such, the use of packing and 
unpacking in these works is also similar to our use of snapshots to temporarily 
escape a given preorder in \autoref{sec:mst}.

\vspace{-0.025cm}

\paragraph{Other Related Works.} For example, \citepos{BengtsonBFGM11} RCF calculus
takes as primitive a monotonic log of formulas to represent the
distributed state of a protocol; and 
\citepos{ChajedCCKZZ17} Crash
Hoare Logic is designed for reasoning about program correctness in the presence of
failure and recovery. We expect the latter could be used to reason about
the Ariadne protocol (\autoref{sec:case-studies.ariadne}),
although, lacking monotonic state, we expect the proof would have to
be conducted using stateful invariants.

\ch{Francois Dupressoir reminded me that in his PhD thesis
  \cite{Dupressoir13} he also used monotonicity in VCC, with some of
  the monotonicity obligations discharged Coq. VCC seems to support
  something called ``two-state invariants'' for its concurrency
  verification.}
\ch{``The term invariant encompasses two-state predicates for the
  before and after states of a state transition.  In this way,
  invariants serve as the rely conditions in a form of rely-guarantee
  reasoning (see \citep{Jones81}, an early formulation of this concept).''}

\ch{Type-driven Development of Concurrent Communicating Systems, Edwin Brady
  To appear in Computer Science Journal 2017.
  \url{https://eb.host.cs.st-andrews.ac.uk/writings/tdd-conc.pdf}
  State transition systems in the types, for instance session types.
  Related to Ilya's work~\citet{SergeyWT18}}

\ch{Derek brought up one more of his papers\cite{DreyerNRB10} where he
  uses a modal operator to lift a predicate wrt a certain kind of
  state extensions. There is no notion of a user-defined preorder or
  anything like that, so the relation seems rather weak.}

\section{Conclusion and future work}
\label{sec:conclusion}

In this paper, we have provided a practical way to ease the
verification of programs whose state evolves monotonically.
In summary, the main distinctive contributions of our work are:

\begin{itemize}
\item A new, simple, core design for reasoning about stable predicates on
monotonic global state, namely, the \emph{monotonic-state monad} \ls$MST$,
together with its \ls$witness$ and \ls$recall$ actions.

\item A demonstration that despite its simplicity, \ls$MST$ is general
enough to account for diverse forms of monotonic state, capturing both
\emph{language-level} (e.g., a memory model with typed references) and
\emph{application-level} (e.g., ghost distributed state) monotonicity.

\item An application of \ls$MST$ to two substantial case studies:
we provide machine-checked proofs of a secure file-transfer application,
and the recent Ariadne state-continuity protocol.

\item A formalization of the monotonic-state monad
  in the context of a higher-order, dependently typed calculus,
  showing how to support both \emph{intrinsic} proofs when treating \ls$MST$ abstractly,
  and \emph{extrinsic} proofs when carefully revealing the underlying representation of \ls$MST$.
\end{itemize}

These contributions are focused on controlling the
temporal evolution of state in sequential programs, leaving spatial
properties to be handled by existing techniques. As discussed in
\autoref{sec:relatedwork}, some strands of
related work profitably combine reasoning about the temporal evolution
of state with concurrency and various forms of aliasing
control, an area that we hope to investigate in the future.
\da{I feel this comment needs to be made more precise.}
We are also currently working on extending \fstar with indexed effects,
so as to make programming and reasoning about \ls$MST$ computations more
convenient, allowing \ls$MST$ to be indexed in \fstar by
state types, preorders, and booleans.
In addition, while our monotonic references support programming with
multiple \emph{local} preorders, we plan to investigate allowing different
parts of the programs to also modularly make use of different \emph{global}
preorders, \EG by using a graded-monads-style ordered monoidal
structure \citep{Katsumata14} on the preorder index of \ls$MST$. 
Finally, we are also working on making formal the connection 
between the intuitive holds-everywhere-in-the-future reading of 
our \ls$witnessed$ logical capability and the standard $\square$
necessity modality of  modal logics. In particular, we conjecture
that if one were to use a hybrid modal logic in which one can bind
the ``current'' modal world in the assertions, such as \citepos{reed09}
Hybrid LF, $\witnessed{s.\varphi}$ could be expressed 
as the combination of the $\square$ modality and this hybrid binding 
operation, i.e., as $\square (\downarrow{\!}s.\varphi)$. 
We hope that such a modal treatment of \ls$witnessed$ will
also address the duplication of specifications issue
discussed in \autoref{sec:discussion}.





%

\ifanon\else
\begin{acks}                            
  
We are grateful to the anonymous reviewers and Nadia Polikarpova for
suggesting many ways in which to improve our work. We also thank David Baelde, Derek
Dreyer, Gordon Plotkin, Fran\c{c}ois Pottier, and the Project Everest
team for many useful discussions.  Danel Ahman's work was done, in
part, during an internship at Microsoft Research. The work of Danel
Ahman, C\u{a}t\u{a}lin Hri\c{t}cu, and Kenji Maillard is supported, in
part, by the
\grantsponsor{1}{European Research Council}{https://erc.europa.eu/}
under ERC Starting Grant SECOMP (\grantnum{1}{715753}).
\end{acks}
\fi

\bibliographystyle{abbrvnaturl}
\bibliography{fstar}

\end{document}

%% file: minimalpreamble.tex
\newcommand\maybecolor[1]{\color{#1}}

\let\ls\lstinline

\def\Snospace~{\S{}}